\documentclass[journal]{IEEEtran}

\usepackage{tikz}
\usepackage{amsmath}

\usepackage[utf8]{inputenc}
\usepackage{amsmath,amssymb,amsfonts}
\usepackage{graphicx}
\usepackage{xurl} 
\usepackage[hidelinks]{hyperref}
\usepackage{textcomp}
\usepackage{xcolor}
\usepackage{pifont}
\usepackage{tabularx}
\usepackage{threeparttable}
\usepackage[ruled,vlined,linesnumbered]{algorithm2e} 
\usepackage{booktabs} 
\usepackage[normalem]{ulem}
\usepackage{multirow} 
\usepackage{cite}
\usepackage{makecell}
\usepackage{amsthm}
\DeclareMathOperator{\Enc}{Enc}
\DeclareMathOperator{\Adv}{Adv}
\DeclareMathOperator{\View}{View}
\DeclareMathOperator{\Sim}{Sim}        
\DeclareMathOperator{\negl}{negl}    

\usepackage{array}  
\newcommand{\vect}[1]{\mathbf{#1}}
\newcommand{\algofont}[1]{\textsf{#1}}
\newcommand{\ct}{\textit{ct}}
\newcommand{\ID}{\textit{ID}}

\newcommand{\fullcircle}{\tikz[baseline=-0.6ex]\fill (0,0) circle (0.6ex);}
\newcommand{\emptycircle}{\tikz[baseline=-0.6ex]\draw (0,0) circle (0.6ex);}
\newcommand{\halfcircle}{\tikz[baseline=-0.6ex]{\draw (0,0) circle (0.6ex);\begin{scope}\clip (0,0) circle (0.6ex);\fill (-0.6ex,-0.6ex) rectangle (0,0.6ex);\end{scope}}}

\theoremstyle{plain}
\newtheorem{theorem}{Theorem}[section]
\newtheorem{lemma}[theorem]{Lemma}
\newtheorem{corollary}[theorem]{Corollary}
\theoremstyle{definition}
\newtheorem{definition}[theorem]{Definition}
\theoremstyle{remark}
\newtheorem{remark}[theorem]{Remark}

\begin{document}

\title{PRAG: End-to-End Privacy-Preserving Retrieval-Augmented Generation}
\author{Zhijun Li, Minghui Xu, Member, IEEE, Huayi Qi, Wenxuan Yu,  Tingchuang Zhang, Qiao Zhang, GuangYong Shang, Zhen Ma, Xiuzhen Cheng, Fellow, IEEE%
\thanks{\spaceskip=.3333em\xspaceskip=.3333em\relax Zhijun Li, Minghui Xu, Wenxuan Yu, Tingchuang Zhang, Qiao Zhang, and Xiuzhen Cheng are with Shandong University. (e-mail: richikun2014@gmail.com, mhxu@sdu.edu.cn, 202115114@mail.sdu.edu.cn, tczhang@mail.sdu.edu.cn, qiao.zhang@sdu.edu.cn, xzcheng@sdu.edu.cn).}%
\thanks{\spaceskip=.3333em\xspaceskip=.3333em\relax Huayi Qi is with Tsinghua University. (e-mail: huayiqi@tsinghua.edu.cn).}%
\thanks{\spaceskip=.3333em\xspaceskip=.3333em\relax GuangYong Shang and Zhen Ma are with Inspur Yunzhou Industrial Internet Co., Ltd. (e-mail: shangguangyong@inspur.com, mazhenrj@inspur.com).}%

\thanks{Corresponding author: Minghui Xu}}


\maketitle

This work has been submitted to the IEEE for possible publication. 
Copyright may be transferred without notice, after which this version 
may no longer be accessible.

\begin{abstract}

Retrieval-Augmented Generation (RAG) is essential for enhancing Large Language Models (LLMs) with external knowledge, but its reliance on cloud environments exposes sensitive data to privacy risks. Existing privacy-preserving solutions often sacrifice retrieval quality due to noise injection or only provide partial encryption. We propose PRAG, an end-to-end privacy-preserving RAG system that achieves end-to-end confidentiality for both documents and queries without sacrificing the scalability of cloud-hosted RAG. PRAG features a dual-mode architecture: a non-interactive PRAG-I utilizes homomorphic-friendly approximations for low-latency retrieval, while an interactive PRAG-II leverages client assistance to match the accuracy of non-private RAG. To ensure robust semantic ordering, we introduce Operation-Error Estimation (OEE), a mechanism that stabilizes ranking against homomorphic noise. Experiments on large-scale datasets demonstrate that PRAG achieves competitive recall (72.45\%–74.45\%), practical retrieval latency, and strong resilience against graph reconstruction attacks while maintaining end-to-end confidentiality. This work confirms the feasibility of secure, high-performance RAG at scale.
\end{abstract}

\begin{IEEEkeywords}
Privacy-preserving RAG, homomorphic encryption, secure similarity retrieval, access-pattern leakage, encrypted HNSW.
\end{IEEEkeywords}


\section{Introduction}
With the rapid proliferation of Large Language Models (LLMs) in question-answering and decision-support settings, combined with recent advances in knowledge-intensive artificial intelligence systems, an increasing number of users are investigating the use of Retrieval-Augmented Generation (RAG) to construct private knowledge-base applications over proprietary data~\cite{khan2025evidencebot,wei2025privacy}. In particular, RAG facilitates the integration of LLMs with intelligent retrieval methods, thereby enabling the deployment of knowledge-centric applications in a flexible, rapid, and low-code fashion~\cite{xu2024generative,zhao2025rag,yu2025ekrag}.

Motivated by these advantages, cloud-based RAG has emerged as the predominant deployment paradigm~\cite{weerasekara2025privacy,cheng2025remoterag}. Major cloud providers, including Amazon, Google, and Alibaba Cloud, now offer native capabilities for seamlessly integrating personal and enterprise document repositories into RAG pipelines~\cite{aws_bedrock_rag, google_vertex_rag, aliyun_model_studio}. 
This allows organizations to rapidly deploy AI agents with scalable, low-latency access to proprietary knowledge, but also causes sensitive document and retrieval indices to reside persistently on cloud servers, exposing a new attack surface. In practice, authorized users may even access such cloud-hosted knowledge bases from mobile devices, further extending both the convenience and the exposure surface of this deployment model. A series of large-scale data breaches since 2024 demonstrates the vulnerability of cloud-hosted data at scale. Notably, the Mother of All Breaches (MOAB) exposed over 26 billion credential records~\cite{cybernews_moab_2024}, while other incidents targeted universities, government contractors, and major technology vendors, including Huawei and Google~\cite{wsu_cyber_update_2025, cybernews_huawei_2025, forbes_google_leak_2025, bleeping_sedgwick_breach_2026}.
These events demonstrate that storing sensitive data in the cloud without strong cryptographic protection poses a systemic and ongoing privacy risk.

\begin{table}[!ht]
\centering
\scriptsize
\begin{threeparttable} 
\setlength{\tabcolsep}{4pt}
\caption{Comparison of Privacy-Preserving RAG Schemes}
\label{tab:comparison}
\begin{tabularx}{\columnwidth}{@{}l 
    >{\hsize=0.5\hsize\raggedright\arraybackslash}X 
    >{\hsize=0.8\hsize\raggedright\arraybackslash}X 
    >{\hsize=1.0\hsize\raggedright\arraybackslash}X 
    >{\hsize=1.6\hsize\raggedright\arraybackslash}X @{}}
\toprule
Scheme & Method & Doc ($\mathcal{D}$) Protection & Query ($\mathcal{Q}$) Protection & Provable Guarantee \\
\midrule
RemoteRAG~\cite{cheng2025remoterag}
& DP\tnote{1}
& Plaintext
& DP-perturbed
& $\varepsilon$-DP ($\mathcal{Q}$) \\
\midrule
\makecell[l]{DP-RAG~\cite{grislain2025rag},\\ DPVoteRAG~\cite{koga2024privacy},\\ Text-DP~\cite{yu2024textual}}
& DP
& DP-perturbed 
& Plaintext 
& $\varepsilon$-DP  ($\mathcal{D}$)\\
\midrule
\makecell[l]{SANNS~\cite{chen2020sanns},\\ TipToe~\cite{Henzinger2023PrivateWS}}
& HE\tnote{2}
& Plaintext
& Encrypted
& Semantic security ($\mathcal{Q}$) \\
\midrule
Pacmann~\cite{Zhou2024PacmannEP}
& PIR\tnote{3}
& Plaintext
& PIR-protected
& PIR security  ($\mathcal{Q}$)\\
\midrule
Fortify~\cite{chrapek2024fortify} & TEE & Plaintext & TEE-protected & Trusted hardware ($\mathcal{Q}$)\\
\midrule
\textbf{PRAG (Ours)} & HE
& Encrypted & Encrypted & Semantic security ($\mathcal{D}, \mathcal{Q}$)\\
\bottomrule
\end{tabularx}

\begin{tablenotes}
    \item[1] DP: Differential Privacy.
    \item[2] HE: Homomorphic Encryption.
    \item[3] PIR: Private Information Retrieval.
\end{tablenotes}
\end{threeparttable}
\end{table}

Current research into privacy-preserving RAG explores a spectrum of solutions ranging from local deployment to statistical noise injection and hardware-based execution environments, each offering different trade-offs between security and utility as shown in Table~\ref{tab:comparison}. To avoid exposing documents and queries to the cloud, local deployments~\cite{wei2025privacy,weerasekara2025privacy,khan2025evidencebot} keep all data on the client side. However, this design is ill-suited for shared, cloud-hosted RAG settings because it lacks support for multi-user collaboration, centralized administration, and scalable management of large-scale knowledge bases.

To address these limitations while remaining in the cloud, some systems adopt Differential Privacy (DP) to reduce information leakage by injecting noise into documents~\cite{koga2024privacy,yu2024textual,grislain2025rag}. While these approaches protect statistical privacy, they do not provide provable data confidentiality; the cloud still stores perturbed document representations in plaintext and can observe access patterns during retrieval. RemoteRAG~\cite{cheng2025remoterag} improves upon this by perturbing user queries, preventing the cloud from directly observing query contents. Nevertheless, the underlying documents remain stored and processed in plaintext, leaving the hosted knowledge base visible to the cloud provider; SANNS~\cite{chen2020sanns}, TipToe~\cite{Henzinger2023PrivateWS}, and Pacmann~\cite{Zhou2024PacmannEP} have the same limitation in that they protect query privacy but not document privacy, which may be acceptable for public dataset but not for private knowledge bases. Alternatively, TEE-based RAG systems~\cite{chrapek2024fortify} protect query processing by executing retrieval within trusted hardware enclaves. This significantly reduces data exposure during computation, though its security remains tied to specific hardware trust assumptions. This raises a new problem:

\begin{center}
    \textit{Is it possible to achieve end-to-end confidentiality for both documents and queries without sacrificing the scalability of cloud-hosted RAG?}
\end{center}


A practical solution must maintain the corpus in the cloud while providing robust cryptographic protection and preserving the full functionality required by RAG. In particular, a privacy-preserving RAG system must support three core operations: \textit{secure similarity computation} between a query embedding and document embeddings, \textit{efficient ranking} to select the top-$k$ contexts for generation, and \textit{real-time updates} to keep the knowledge base synchronized with continuously changing data. Approximate Nearest Neighbor (ANN) search, with its sub-linear query complexity and favorable recall-latency trade-off, is naturally suited for retrieval over large-scale embedding corpora in modern RAG systems. To ensure the cloud learns nothing about document contents, these operations cannot be performed on plaintext and must instead be executed entirely in the ciphertext domain. Under this end-to-end encrypted cloud setting, these requirements effectively leave \textit{Secure Similarity Retrieval (SSR)}~\cite{huang2019toward,zheng2020achieving,chen2020sanns,wang2024secure,tan2021efficient} as the most viable design direction, since it enables similarity retrieval directly over encrypted data in the cloud. However, existing Secure Similarity Retrieval schemes still have limitations in practical RAG deployments.

\subsection{Why SSR Fails in Privacy-Preserving RAG}

\textbf{Partial Support for RAG Functionality in Secure Similarity Retrieval.} 
Although \textit{secure similarity computation}, \textit{efficient ranking}, and \textit{real-time updates} are the pillars of privacy-preserving RAG systems, current SSR schemes fail to provide these three capabilities concurrently, as shown in Table~\ref{tab:SSRcomparison}. The following analysis identifies the specific gaps in existing solutions.

First, several schemes support similarity computation but rely on encryption primitives that have been shown to be insecure under practical attack models, such as ASPE-based constructions~\cite{huang2019toward,zheng2020achieving}. Although these approaches support encrypted similarity computation, their security guarantees are limited for protecting knowledge bases~\cite{li2019insecurity}. 
Second, efficient RAG retrieval requires ranking similarity scores to select the most relevant context. However, schemes such as SPE-Sim~\cite{zheng2020achieving} and SESR~\cite{wang2024secure} do not support numeric comparison in the ciphertext domain. Without native encrypted comparison, these systems cannot perform top-$k$ ranking without leaking sensitive score distributions or requiring trusted-party decryption.
Third, the dynamic nature of RAG requires support for evolving datasets, yet many schemes, such as PDQ~\cite{tan2021efficient} and SESR~\cite{wang2024secure}, lack real-time update capabilities. Their index structures often rely on global organization or rigid partitioning, making new insertions computationally prohibitive. This limitation affects the handling of continuously evolving knowledge bases. 

\begin{table}[!t]
\footnotesize
\centering
\begin{threeparttable}
    \caption{Comparison of Secure Similarity Retrieval Schemes}
    \label{tab:SSRcomparison}
    \begin{tabularx}{\columnwidth}{l>{\hsize=1.3\hsize}X>{\hsize=0.8\hsize}X>{\hsize=0.8\hsize}X}
        \toprule
        Scheme & Secure Similarity Computation & Efficient Ranking & Real-Time Updates \\
        \midrule
        Huang et al.~\cite{huang2019toward} & \emptycircle & \fullcircle & \halfcircle\\
        SPE-Sim~\cite{zheng2020achieving}   & \emptycircle & \emptycircle & \fullcircle \\
        SESR~\cite{wang2024secure}          & \fullcircle & \emptycircle & \emptycircle \\
        PDQ~\cite{tan2021efficient}         & \fullcircle & \halfcircle & \emptycircle \\
        MSecKNN~\cite{Li2025MSecKNNMS}          & \fullcircle & \halfcircle & \fullcircle \\
        FSkNN~\cite{Fukuchi2025SecureKF}          & \fullcircle & \halfcircle & \fullcircle \\
        \midrule
        \textbf{PRAG (Ours)}               & \fullcircle & \fullcircle & \fullcircle \\
        \bottomrule
    \end{tabularx}
    \begin{tablenotes}
        \small
        \item[] \fullcircle: Supported; \emptycircle: Not supported; \halfcircle: Partial support under specific settings or with performance constraints.
    \end{tablenotes}
\end{threeparttable}
\end{table}

\textbf{Distortion of retrieval ranking.}
      Even when encrypted comparison is available, or when the system relies on an HE-friendly ranking surrogate, the ordering of similarity scores can change under CKKS-style approximate arithmetic. In privacy-preserving RAG, each similarity score is computed through multiple steps, including packed inner products, rescaling and rounding, and in fully homomorphic settings, polynomial approximations for selection. These steps introduce errors that may vary across candidates and accumulate during retrieval~\cite{Lee2021HighPrecisionBO,bae2024bootstrapping}.

      This leads to two challenges. First, top-$k$ retrieval is sensitive to score differences~\cite{liu2023lost,yu2024defense}. If two candidates have similar true scores, even a small error can flip their order. Second, error depends on both input data and retrieval path. Different candidates may accumulate different error due to different operation counts, such as different traversal lengths, making a single absolute error bound insufficient. The core challenge is therefore to keep total perturbation below the similarity margin that separates relevant documents from near ties, so encrypted retrieval preserves ranking with high probability.

\subsection{Our Contribution}
To address these challenges, we propose PRAG, an end-to-end privacy-preserving RAG framework that enables high-accuracy retrieval over encrypted knowledge bases. Our contributions are summarized as follows:
\begin{itemize} 
\item \textbf{Bridging secure similarity retrieval and basic RAG functionality.} To address the incompatibility between existing secure similarity retrieval schemes and basic RAG requirements, we develop a CKKS-compatible encrypted retrieval backend that supports (1) similarity scoring via homomorphic inner products, (2) ranking through Chebyshev polynomial approximations or client-assisted comparisons, and (3) real-time updates using a ciphertext-resident hierarchical index with homomorphic clustering and encrypted HNSW navigation.

\item \textbf{An end-to-end privacy-preserving RAG system with dual-mode operation.} We propose PRAG, an end-to-end privacy-preserving RAG system that supports two deployable modes within a unified encrypted architecture: a non-interactive PRAG-\uppercase\expandafter{\romannumeral1} mode that performs retrieval fully in the ciphertext domain via homomorphic-friendly approximations, and an interactive, client-assisted PRAG-\uppercase\expandafter{\romannumeral2} mode that leverages limited client-side assistance to improve ranking precision and retrieval accuracy, achieving average top-10 retrieval latencies of under 1.29 seconds and around 7.91 seconds, respectively, on a 100,000-sample dataset. In our experiments, the two modes achieve recall rates of 72.45\% and 74.45\%, respectively, while sharing the same encrypted corpus and index and enabling on-demand switching without re-encrypting data or rebuilding indices. The system further incorporates access-pattern mitigation against graph reconstruction attacks, together with formal leakage analysis and empirical evaluation of the resulting security-utility trade-off.

\item \textbf{Noise-aware ranking stabilization via Operation-Error Estimation (OEE).} We introduce OEE, a noise control method that models how approximation error and CKKS noise affect similarity scores during encrypted retrieval. Rather than focusing only on numeric accuracy, OEE targets preservation of the relative order of similarity scores, which significantly impacts retrieval quality in RAG. Specifically, OEE accounts for: (1) polynomially approximated comparisons, which occur only in PRAG-\uppercase\expandafter{\romannumeral1}; and (2) homomorphic noise and hierarchical index traversal, which affect both PRAG-\uppercase\expandafter{\romannumeral1} and PRAG-\uppercase\expandafter{\romannumeral2}. Using OEE, we guide parameter selection to bound ranking errors, making ranking stability an explicit design goal for encrypted retrieval.
\end{itemize}

The remainder of this paper details the system architecture, security model,
and protocol design, followed by experimental evaluation and discussion.

\section{Related Work}

\subsection{Privacy-Preserving RAG (PRAG)}
Existing privacy-preserving RAG research can be broadly organized into three routes: statistical perturbation of the retrieval pipeline, system-level trust decomposition, and selective plaintext assistance during retrieval or ranking. This taxonomy clarifies the central gap in the literature: prior systems usually strengthen one aspect of privacy, but they do not simultaneously provide end-to-end corpus confidentiality, accurate semantic retrieval, and practical cloud deployment.

RAG first introduced by Lewis et al.~\cite{Lewis2020} and later refined by Active RAG~\cite{Jiang2023} established the retrieval-generation paradigm, but these works do not consider the privacy risks of sensitive corpora. Subsequent analyses by Huang et al.~\cite{huang2023privacy} and Zeng et al.~\cite{zeng2024good} showed that retrieved knowledge and training signals can leak private information. One line of defense therefore adds statistical noise to the retrieval pipeline. DP-based approaches~\cite{koga2024privacy,yu2024textual,grislain2025rag} and embedding perturbation methods such as PRESS~\cite{he2025press} improve privacy for queries or outputs, but they inevitably distort similarity relationships and still leave the outsourced corpus exposed in plaintext.

A second route reduces reliance on a single cloud server through architectural redesign. Federated systems such as FRAG~\cite{zhao2024frag}, blockchain-based systems such as D-RAG~\cite{e_andersen2025d}, split execution~\cite{wei2025privacy}, and isolated-enclave solutions~\cite{weerasekara2025privacy,khan2025evidencebot} all limit the trust placed in one component of the system. However, these designs typically introduce substantial communication or coordination overhead and still do not provide an efficient encrypted vector index that supports RAG-style semantic search.

A third route keeps part of the pipeline encrypted but reveals intermediate information to recover accuracy or efficiency. Privacy-Aware RAG~\cite{zhou2025privacy}, for example, requires selective decryption during ranking, which breaks end-to-end confidentiality and reintroduces exposure during the most ranking-sensitive stage of retrieval.

In summary, the three main RAG directions correspond to three incomplete trade-offs: perturbation-based methods sacrifice utility, trust-decomposition methods sacrifice practicality, and plaintext-assisted methods sacrifice full confidentiality. PRAG is designed to close this gap by supporting clustering, HNSW traversal, and ranking over a fully encrypted corpus under an untrusted server model.

\subsection{Secure Similarity Retrieval }
Existing secure similarity retrieval (SSR) work can likewise be summarized along three routes: query-private retrieval over plaintext corpora, encrypted similarity computation with weak or limited ranking semantics, and structured encrypted indexes that still fall short of RAG's ANN requirements. This lens makes clear why current SSR primitives cannot be plugged into RAG directly without losing either security or functionality.

Early works such as SANNS~\cite{chen2020sanns} and TipToe~\cite{Henzinger2023PrivateWS} protect the query while leaving the document embeddings exposed on the server. Pacmann~\cite{Zhou2024PacmannEP} shifts crucial computation back to the client because it cannot support server-side comparison and similarity computation securely. These systems therefore do not meet the end-to-end confidentiality requirement of outsourced RAG.

Huang et al.~\cite{huang2019toward} proposed an ASPE-based method, but ASPE is not secure against ciphertext-only attacks~\cite{li2019insecurity}. Subsequent efforts extended query semantics to support set similarities, as in SPE-Sim~\cite{zheng2020achieving} and SESR~\cite{wang2024secure}, or numerical comparisons over integers, as in PDQ~\cite{tan2021efficient}. Even when these systems support secure arithmetic, they do not provide the continuous, ranking-preserving comparisons needed for semantic vector retrieval in RAG.

Some methods~\cite{xia2022format} use clustering or quantization, but they partition the feature space rigidly and break semantic continuity across clusters. Real-time update mechanisms from keyword-based systems~\cite{stefanov2014practical, etemad2018efficient, dou2024dynamic} are not directly compatible with continuous embedding retrieval. Recent designs based on R-trees~\cite{li2023enabling,wang2023efficient} or specialized encodings~\cite{sun2023soar,lee2025bit} improve structure, yet none jointly support ANN graph navigation, iterative centroid refinement, and semantically faithful vector operations.

Overall, existing SSR schemes each cover only a subset of the requirements for encrypted RAG. For example, SANNS~\cite{chen2020sanns} requires multiple online parties, and update operations can disrupt cluster balance and trigger costly reconstruction. As a result, no prior SSR scheme simultaneously satisfies the navigational, semantic, dynamic, and cryptographic requirements of practical privacy-preserving RAG.

\section{Preliminaries}

\subsection{Graph-Based ANN Search}
 Approximate Nearest Neighbor (ANN) search is a fundamental operation in high-dimensional vector databases that trades exactness for efficiency, returning approximate nearest neighbors with sub-linear query complexity. Among ANN approaches, graph-based methods have emerged as the state-of-the-art for high-dimensional similarity search due to their superior recall-latency trade-off. These methods construct a proximity graph over the dataset where each node represents a data point and edges connect nearby points. Search proceeds by greedily traversing the graph from an entry point, iteratively moving to neighbors closer to the query until reaching a local optimum. Compared to tree-based methods (e.g., KD-trees) that suffer from the curse of dimensionality, and hashing-based methods (e.g., LSH) that require parameter tuning for specific distance distributions, graph-based methods adaptively capture the intrinsic data manifold and achieve robust performance across diverse datasets.

Hierarchical Navigable Small World (HNSW) is a representative graph-based ANN algorithm that extends the basic proximity graph with a multi-layered hierarchical structure. It organizes data points into a multi-layered graph to perform search in a coarse-to-fine manner, where sparse upper layers facilitate rapid global navigation and dense lower layers enable accurate local refinement. At each layer, the search follows a greedy routing rule: starting from a designated entry node, the algorithm iteratively evaluates the distances to the current node's neighbors and moves to the neighbor that is closest to the query vector. This process continues until no adjacent neighbor is closer to the query than the current node, reaching a local optimum that serves as the entry point for the next lower layer.

To illustrate this mechanism concretely, consider a typical three-layer HNSW configuration consisting of a sparse top layer ($L_2$), a middle layer ($L_1$), and a dense base layer ($L_0$) containing all data points, as illustrated in Fig.~\ref{hnsw}. The retrieval process is initialized at a predefined global entry node in $L_2$, where the algorithm performs a rapid greedy traversal to "zoom in" on the general region of the query while bypassing large portions of irrelevant data. Once a local optimum is identified in $L_2$, the search transitions to $L_1$, using that optimum as the new entry point. Since $L_1$ possesses a denser topology, it provides a finer-grained routing path to further narrow the search space. Finally, the search descends into the densest layer, $L_0$, to execute a precise localized greedy search, often maintaining a dynamic candidate list to pinpoint the exact top-$k$ nearest neighbors. This hierarchical transition ensures that the search complexity scales sub-linearly, typically $O(\log N)$, while maintaining high recall in high-dimensional spaces.
\begin{figure}[!ht]
	\centering
	\includegraphics[width=0.6\columnwidth]{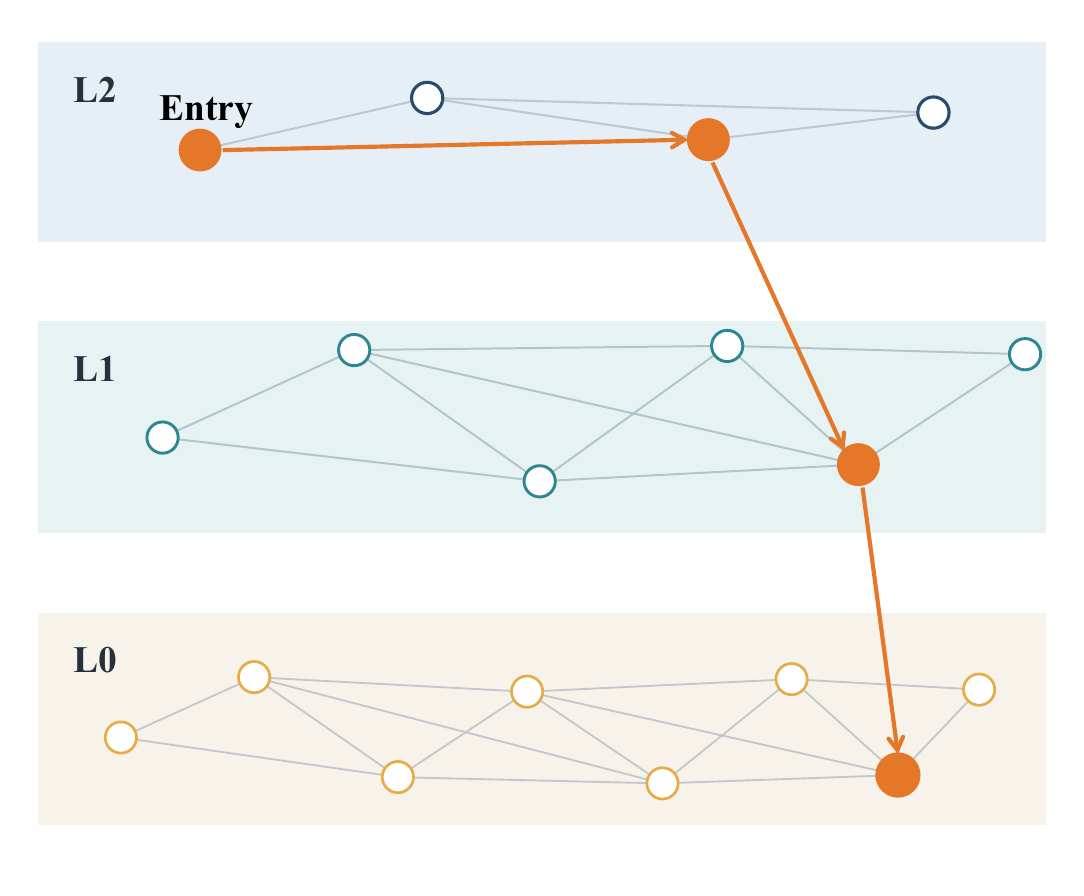}
	\caption{Schematic diagram of the HNSW structure.}   
  \label{hnsw}
\end{figure}

\subsection{K-means Clustering}
K-means clustering is a widely used unsupervised learning algorithm for partitioning $\vert \mathcal{D} \vert$ data points into $k$ clusters, where each data point belongs to the cluster with the nearest centroid. The algorithm iteratively minimizes the within-cluster sum of squared distances by updating cluster assignments and centroids until convergence. Formally, given a set of data points $\{x_1, x_2, \dots, x_n\}$ in a $d$-dimensional space, the goal is to find centroids $\{\mu_1, \dots, \mu_k\}$ that minimize
\[
\sum_{i=1}^{n} \min_{j \in \{1,\dots,k\}} \|x_i - \mu_j\|^2.
\]

\subsection{CKKS and Primitives}
We briefly introduce the CKKS homomorphic encryption abstraction used in this work. CKKS supports approximate arithmetic on packed vectors, enabling homomorphic
addition, multiplication, and slot rotations. We use the following high-level primitives:
\begin{itemize}
  \item $\mathsf{KeyGen}(\lambda)\to(\mathit{pk},\mathit{sk},\mathit{evk})$.
  \item $\mathsf{Encrypt}(\mathit{pk},\mathbf{v})\to\ct_{\mathbf{v}}$ and
        $\mathsf{Decrypt}(\mathit{sk},\ct)\to\mathbf{v}'$.
  \item $\mathsf{Add},\ \mathsf{Mul},\ \mathsf{Rotate}$: homomorphic slot-wise
        operations.
\end{itemize}

\subsection{Chebyshev Polynomial Approximation (ChebyApprox)}
Let $\vect{s} = (s_1, \dots, s_C) \in \mathbb{R}^C$ be similarity scores or distance scores, and let $\tau > 0$ be a scaling parameter. Since CKKS supports only additions and multiplications efficiently, we approximate a monotone comparison surrogate over a bounded interval with a degree-$d$ Chebyshev polynomial
\[
C_d(x) = \sum_{m=0}^{d} a_m T_m(x),
\]
where $T_m(\cdot)$ is the $m$-th Chebyshev polynomial of the first kind. We then define
\[
\algofont{ChebyApprox}(s_i) := w_i = \frac{C_d(-s_i / \tau)}{\sum_{j=1}^{C} C_d(-s_j / \tau)}, \quad i \in [C].
\]
The resulting weights $\vect{w} = (w_1, \dots, w_C)$ satisfy $\sum_{i=1}^{C} w_i = 1$ and provide an HE-friendly surrogate for comparison-based ranking.


\section{Problem Statement}
\label{sec:system_design}




\subsection{System Setting}
Our proposed framework is designed for an environment where data providers can contribute to a shared encrypted RAG database. The architecture comprises two main entities: a trusted Client, and a semi-honest Cloud Server. 

A trusted Client serves two roles. As the data owner, it normalizes private documents into embedding vectors, builds the search index, encrypts the entire dataset and index structure, and uploads the ciphertexts to the Cloud Server. As an authorized user, it transforms a query into multiple sub-queries, encrypts them, and submits them to the server; upon receiving the encrypted retrieval results from the cloud, it decrypts them and forwards the plaintext context to the large language model for response generation.

The Cloud Server is a semi-honest provider that hosts the encrypted database. Upon receiving an encrypted query, it performs all similarity-matching computations directly on ciphertexts, without any decryption key access, and returns the top-$k$ encrypted candidates to the Client. The Client then decrypts and reranks the results locally.

\subsection{Threat Model}
\label{sec:threat_model}

We model the cloud server as a Probabilistic Polynomial-Time (PPT) adversary $\mathcal{A}$ under the semi-honest model. The adversary:
\begin{itemize}
    \item Infers sensitive information from all observable data;
    \item Holds no secret key $\mathit{sk}$ and does not collude with others.
    \item Correctly executes all protocols using provided keys; 
\end{itemize}

The adversary's view consists of:
\begin{equation}
\mathsf{View}_{\mathcal{A}} = \left(\{\ct_{\mathbf{v}}^{(i)}\}_{i=1}^{|\mathcal{D}|}, \{\ct_{\mathbf{q}}^{(j)}\}_{j=1}^{Q}, \mathcal{CT}_{\text{inter}}, \mathsf{Params}, \mathcal{L}\right),
\end{equation}
where $\{\ct_{\mathbf{v}}^{(i)}\}$,$\{\ct_{\mathbf{q}}^{(j)}\}$ are encrypted database, queries, $\mathcal{CT}_{\text{inter}}$ are intermediate ciphertexts generated during computation, $\mathsf{Params}$ are public parameters such as $\lambda$, $N$, and $\tau$, and $\mathcal{L}$ is the leakage function defined in Appendix~\ref{leakage}.

\subsection{Design Goals}
\begin{itemize}

\item \textbf{Data Privacy (Semantic Confidentiality).}
All document embeddings, query embeddings, similarity scores, and final rankings must remain semantically hidden from the cloud server.
Formally, the CKKS-encrypted ciphertexts ensure IND-CPA security: for any two datasets $\mathcal{V}_0, \mathcal{V}_1$ and queries $q_0, q_1$ of equal size, the adversary's views are computationally indistinguishable.
We formalize this as semantic confidentiality.
\begin{definition}[Semantic Confidentiality]
\label{def:semantic_confidentiality}
Our framework achieves semantic confidentiality if for any PPT adversary $\mathcal{A}$, there exists a PPT simulator $\mathcal{S}$ such that:
\begin{equation}
\left\{\mathsf{View}_{\mathcal{A}}(\mathcal{D}, \{\mathbf{q}^{(i)}\})\right\} \stackrel{c}{\approx} \left\{\mathcal{S}\left(\mathcal{L}(\mathcal{D}, \{\mathbf{q}^{(i)}\}), 1^\lambda\right)\right\},
\end{equation}
where $\stackrel{c}{\approx}$ denotes computational indistinguishability.
\end{definition}
Additionally, the permitted leakage $\mathcal{L}$ consists of access patterns during HNSW traversal and cluster probing. PRAG explicitly bounds this leakage via random access perturbation to prevent graph reconstruction attacks. We categorize access-pattern-based graph reconstruction attacks as a class of LAAs, where the adversary exploits observable retrieval traces to infer hidden index topology.

\item \textbf{Ranking Stability under Noise.}
The system ensures that semantic ranking is robust to approximation error and homomorphic noise.
In particular, retrieval quality is preserved by explicitly controlling ranking deviations rather than only minimizing numeric error.

\item \textbf{Practical and Efficient Dual-mode System Design.}
The overall design targets real-world RAG workloads by supporting efficient and accurate retrieval on large-scale datasets with real-time updates, and by enabling dual-mode operation to trade latency for ranking accuracy, including fully homomorphic non-interactive retrieval and client-assisted retrieval.

\end{itemize}

\begin{table}[!ht]
\caption{Summary of Core Notations}
\label{tab:notations}
\centering
\small
\begin{tabularx}{\linewidth}{@{}l X@{}}
\toprule
\textbf{Symbol} & \textbf{Description} \\
\midrule
$\vert \mathcal{D} \vert$ 
& Dataset size \\
$ d$ 
& Vector dimension \\
$C$ 
& Number of clusters \\
$T$ 
& K-means iterations \\
$\tau$ 
& Chebyshev scaling parameter \\

$\ct_{(\cdot)}$ 
& Encrypted ciphertext, such as vector $\ct_{\mathbf{v}}$ or score $\ct_{s}$ \\
$\mathcal{N}$ 
& Set of neighbors \\
$C_{\mathsf{probe}}$ 
& Number of clusters probed per query \\
$k_{\mathsf{sub}}$ 
& Number of sub-queries generated per cluster\\

$\mathcal{S}_{\text{IP}}$ 
& Noise of a single homomorphic inner product \\
$\mathbf{v}, \mathbf{q}$ 
& Plaintext vectors (data and query) \\
$N, \lambda$ 
& CKKS polynomial modulus degree and security parameter \\
$\mathit{pk}, \mathit{sk}, \mathit{evk}$ 
& Public key, secret key, and evaluation keys ($\mathit{rlk}, \mathit{galk}$) \\
$L, M, \text{Dep}_h$ 
& HNSW parameters: max layers, max degree, and search depth \\
\bottomrule
\end{tabularx}
\end{table}

\section{Privacy-Preserving RAG}
\label{sec:framework_and_algorithms}
\begin{figure}[!ht]
	\centering
	\includegraphics[width=\columnwidth]{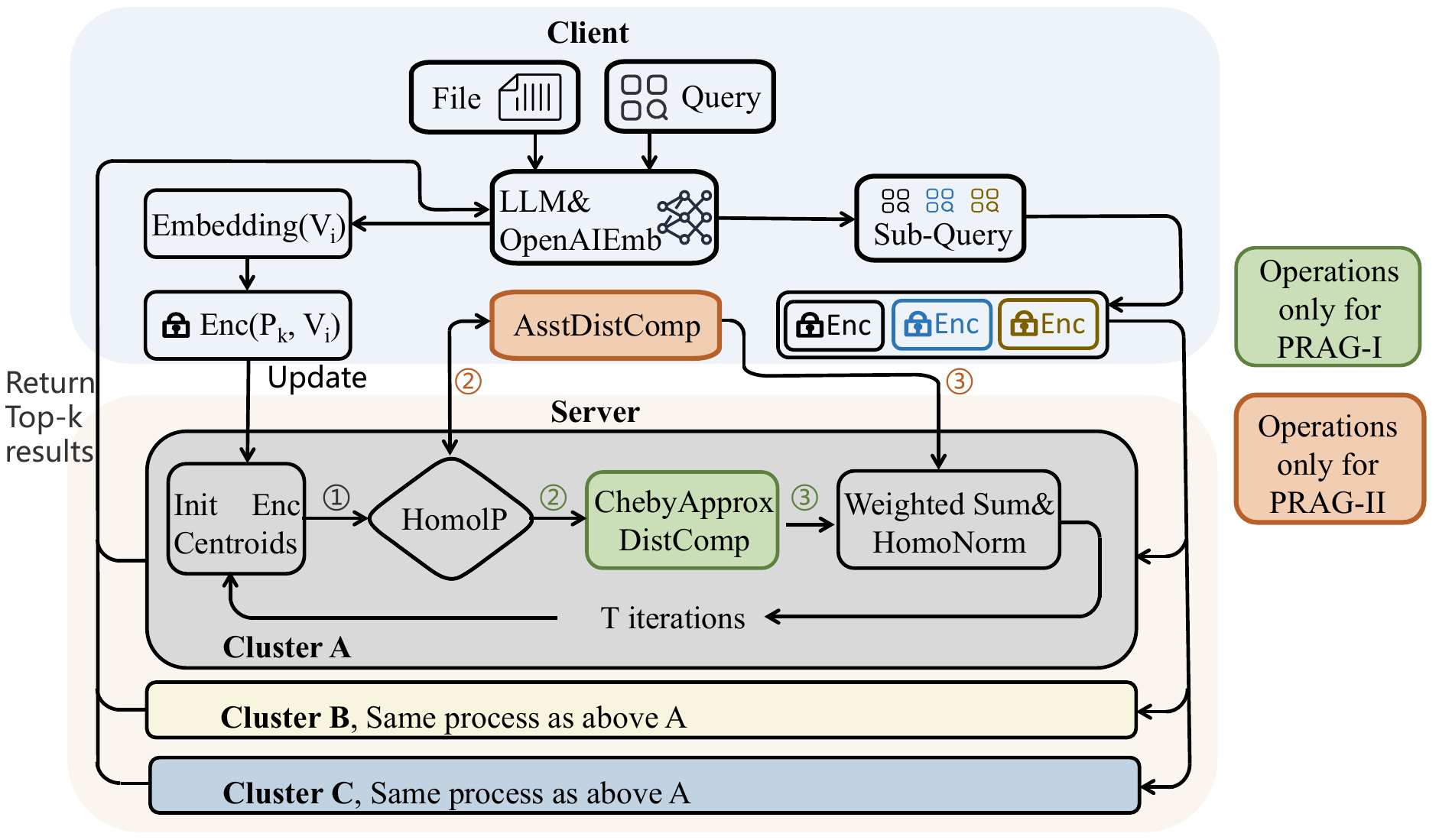}
	\caption{Architecture of PRAG-\uppercase\expandafter{\romannumeral1} and PRAG-\uppercase\expandafter{\romannumeral2}. Note: Numbers in green circles denote steps exclusive to PRAG-\uppercase\expandafter{\romannumeral1}, while numbers in brick-red circles indicate steps specific to PRAG-\uppercase\expandafter{\romannumeral2}; the remaining steps are common to both schemes.}   
  \label{system}
\end{figure}
In this section, we present two complementary privacy-preserving RAG frameworks as shown in Fig.~\ref{system}: 
a non-interactive PRAG-\uppercase\expandafter{\romannumeral1} that achieves end-to-end encrypted retrieval,
and an interactive PRAG-\uppercase\expandafter{\romannumeral2} that enhances ranking precision.
Both frameworks share the same data preprocessing and response generation stages, but differ in Setup, Retrieve and Update stages. The key notations used in our protocols are shown in Table~\ref{tab:notations}.

\subsection{\texorpdfstring{PRAG-\uppercase\expandafter{\romannumeral1}}{PRAG-I}}
\label{subsec:chebyshev_based}
The PRAG-\uppercase\expandafter{\romannumeral1} framework performs retrieval homomorphically on the cloud using CKKS, and uses a Chebyshev polynomial approximation for ranking. 

\subsubsection{Detailed Protocols}\label{detail}
We present the pseudocode for Protocols $\Pi_{\mathsf{Setup}}$, $\Pi_{\mathsf{Retrieval}}$, and $\Pi_{\mathsf{Update}}$ in Section~\ref{detail}, and Protocol~$\Pi_{\mathsf{LocalRAG}}$ in Appendix~\ref{algorithm}.

\par \textbf{Setup Stage: Encrypted Upload and Index Construction.}

As shown in Protocol~\ref{setup}, the client encrypts each vector $\mathbf{v}_i \in \mathcal{V}$ with its identifier $\ID_i$ under $\mathit{pk}$ to obtain ciphertexts $\ct_i$. CKKS parameters, including the polynomial modulus degree $N$, the moduli chain $q_i$, and the initial scaling $\Delta$, are chosen to match the multiplicative depth of the homomorphic inner-product circuit. The client then uploads $\{\ct_i\}$ and evaluation keys $\mathit{evk}$ to the untrusted server, which stores them as a two-level encrypted index $\mathcal{I}$.

\begin{algorithm}[!ht]
\SetAlgoLined
\SetKwInOut{KwIn}{KwIn}
\SetKwInOut{KwOut}{KwOut}
\footnotesize
\caption{Protocol $\Pi_{\mathsf{Setup}}$}
\label{setup}
\KwIn{Plaintext vectors $\mathcal{V}$, public key $\mathit{pk}$, evaluation keys $\mathit{evk}$, parameters $C, T, \tau$.}
\KwOut{Encrypted index $\mathcal{I}$ at the server and encrypted centroids $\{\ct_{\mu_j}\}$ returned to the client.}

\smallskip
\tcp{Client-Side Encryption}
$\mathcal{CT}_{\mathcal{V}} \gets \{(\ID_i, \mathsf{Enc}(\mathit{pk}, \mathbf{v}_i)) \mid (\ID_i, \mathbf{v}_i)\in\mathcal{V}\}$\;
Upload $\mathcal{CT}_{\mathcal{V}}$ and $\mathit{evk}$ to the server\;

\smallskip
\tcp{Level-1: Encrypted Chebyshev-Weighted K-means Clustering}
Randomly initialize encrypted centroids $\{\ct_{\mu_1},\dots,\ct_{\mu_C}\}$\;
\For{$t \gets 1$ \KwTo $T$}{
  Initialize $\mathcal{P}_j.\mathsf{sum}\gets 0$ and $\mathcal{P}_j.\mathsf{count}\gets 0$ for all $j\in[C]$\;
  \ForEach{$(\ID_i,\ct_i) \in \mathcal{CT}_{\mathcal{V}}$}{
    Compute $\ct_{s_j}\gets \mathsf{HomoIP}(\ct_i,\ct_{\mu_j})$ for all $j$\;
    Compute $\ct_{w_j}\gets \mathsf{ChebyApprox}(\{\ct_{s_j}\},\tau)$\;
    Update $\mathcal{P}_j.\mathsf{sum}\gets \mathcal{P}_j.\mathsf{sum}+\ct_i\cdot\ct_{w_j}$ and $\mathcal{P}_j.\mathsf{count}\gets \mathcal{P}_j.\mathsf{count}+\ct_{w_j}$\;
  }
  \For{$j \gets 1$ \KwTo $C$}{
    \If{$\mathcal{P}_j.\mathsf{count}\neq 0$}{
      $\ct_{\mu_j}\gets \mathsf{HomoNorm}(\mathcal{P}_j.\mathsf{sum},\mathcal{P}_j.\mathsf{count})$\;
    }
  }
}
Store encrypted centroids as $\mathcal{I}_{\mathsf{cluster}}$\;

\smallskip
\tcp{Level-2: Encrypted HNSW Construction}
\For{$j \gets 1$ \KwTo $C$}{
  Construct encrypted HNSW graph $G_j$ using $\mathsf{SecureInsert}$\;
}
$\mathcal{I}\gets(\mathcal{I}_{\mathsf{cluster}},\mathcal{I}_{\mathsf{HNSW}})$\;
Send encrypted centroids $\{\ct_{\mu_j}\}$ to the client and output $\mathcal{I}$\;
\end{algorithm}

The server then runs encrypted K-means. It initializes centroids by sampling $C$ vectors and encrypting them as $\{\ct_{\mu_1}, \dots, \ct_{\mu_C}\}$. For $T$ iterations, it computes homomorphic inner products $\ct_{s_j}$ between each data point $\ct_i$ and all centroids, applies $\algofont{ChebyApprox}$ with scaling parameter $\tau$ to obtain weights $\ct_{w_j}$, accumulates weighted sums in each partition $\mathcal{P}_j$, and updates centroids via $\algofont{HomoNorm}$. Details of $\algofont{HomoNorm}$ are given in Appendix~\ref{algorithm}.

Once clustering is complete, the encrypted centroids form the cluster-level index $\mathcal{I}_{\text{cluster}}$. For the second level, the server constructs an encrypted HNSW graph $G_j$ for each cluster by iteratively inserting assigned encrypted vectors using the secure insertion procedure in Protocol~\ref{protocol:secure_update}. The full index $\mathcal{I}$ combines both levels. Finally, encrypted centroids are returned to the client and decrypted into plaintext cluster metadata $C_{\text{meta}}$, which contains centroids and semantic labels. Because this metadata is aggregate, it does not reveal individual data points and can be safely used for subsequent retrieval.

\par \textbf{Retrieve Stage: Hierarchical and Fused Secure Retrieval.} 
As shown in Protocol~\ref{protocol:secure_retrieval}, retrieval uses a hierarchical index and query transformation while keeping server-side computation efficient. The process starts with client-side query decomposition. The original query $q$ is embedded into $\vect{v}_q$ via OpenAIEmbeddings (OpenAI\_Emb). Using plaintext cluster metadata $C_{\text{meta}}$, which is decrypted from the server during setup, the client computes similarities between $\vect{v}_q$ and cluster centroids, then selects the top-$C_{\text{probe}}$ clusters $\mathcal{J}_{\text{best}}$ for semantic routing without revealing query content to the server.

\begin{algorithm}[t]
\SetAlgoLined
\SetKwInOut{KwIn}{KwIn}
\SetKwInOut{KwOut}{KwOut}
\footnotesize
\caption{Protocol $\Pi_{\mathsf{SecureRetrieval}}$}
\label{protocol:secure_retrieval}
\KwIn{Query $q$, top-$k$ parameter $K$, parameters $(C_{\mathsf{probe}}, k_{\mathsf{sub}})$, client metadata $C_{\mathsf{meta}}$, and scaling parameter $\tau$.}
\KwOut{Aggregated encrypted retrieval results.}

\smallskip
\tcp{Client-Side Query Decomposition}
$\mathbf{v}_q \gets \mathsf{OpenAI\_Emb}(q)$\;
Compute similarity scores between $\mathbf{v}_q$ and cluster centroids using $C_{\mathsf{meta}}$\;
Select top-$C_{\mathsf{probe}}$ clusters and denote them as $\mathcal{J}_{\mathsf{best}}$\;

\smallskip
\tcp{Encrypted Sub-query Generation}
$\mathit{EncSubQueries} \gets \emptyset$\;
\ForEach{$j^\ast \in \mathcal{J}_{\mathsf{best}}$}{
  Retrieve semantic label from $C_{\mathsf{meta}}[j^\ast]$\;
  Generate $k_{\mathsf{sub}}$ sub-queries via $\mathsf{LLM.Decompose}$\;
  \ForEach{$q_{\mathsf{sub}}$ in generated sub-queries}{
    $\mathbf{v}_{\mathsf{sub}} \gets \mathsf{OpenAI\_Emb}(q_{\mathsf{sub}})$\;
    $\ct_{\mathsf{sub}} \gets \mathsf{Enc}(\mathit{pk}, \mathbf{v}_{\mathsf{sub}})$\;
    Append $(\ct_{\mathsf{sub}}, j^\ast)$ to $\mathit{EncSubQueries}$\;
  }
}
Send $\mathit{EncSubQueries}$ and $K$ to the server\;

\smallskip
\tcp{Server-Side Targeted Encrypted Search}
$\mathit{AggregatedResults} \gets \emptyset$\;
\ForEach{$(\ct_{\mathsf{sub}}, j^\ast) \in \mathit{EncSubQueries}$}{
  Retrieve encrypted HNSW graph $G_{j^\ast}$\;
  Perform greedy encrypted ANN search on $G_{j^\ast}$ with $\ct_{\mathsf{sub}}$ and obtain top-$k$ identifiers\;
  Append $(j^\ast, \mathit{PartialResults})$ to $\mathit{AggregatedResults}$\;
}
\Return $\mathit{AggregatedResults}$\;
\end{algorithm}

To improve recall, the client performs query fusion within each selected cluster. For each cluster $j^*$, it retrieves the semantic label from $C_{\text{meta}}$ and uses a local LLM via $\algofont{LLM.Decompose}$ to generate $k_{\mathsf{sub}}$ diverse sub-queries tailored to the cluster topic and original query. Each sub-query is embedded by OpenAI\_Emb, encrypted into $\ct_{sub}$, and paired with its target cluster ID. The encrypted sub-query set is sent to the server with the desired top-$k$ parameter.

On the server side, retrieval is targeted and efficient. For each encrypted sub-query $\ct_{sub}$ and designated cluster $j^*$, the server retrieves the corresponding encrypted HNSW graph $G_{j^*}$ from $\mathcal{I}_{\text{HNSW}}$. It then runs greedy search from entry point $\textit{ep}$ to find top-$k$ nearest neighbors in the encrypted domain using homomorphic distance computations. Partial results containing candidate identifiers and scores are aggregated and returned to the client. Restricting search to selected clusters and parallelizing sub-queries reduces server overhead while preserving homomorphic security. The client then fuses and reranks aggregated results in Appendix~\ref{protocol:LocalRAG} to produce final context.

To reduce access-pattern leakage during retrieval, PRAG-\uppercase\expandafter{\romannumeral1} further augments the above search process with randomized protection mechanisms. First, for each real encrypted sub-query, the client injects a $\rho$-fraction of dummy traversal requests that target randomly selected clusters and trigger random walks of matched depth in the corresponding encrypted HNSW graphs. In addition, each real traversal is padded to a fixed observable length $\ell_{\max}$ by appending random neighbor visits, so the server observes a mixture of real and dummy paths with aligned lengths rather than a clean trace of the true greedy search. Second, the encrypted index is periodically re-encrypted and re-clustered after every epoch of $E$ queries, using fresh ciphertext randomness, fresh K-means initialization, and freshly rebuilt HNSW layers. This prevents the server from accumulating stable long-term traversal statistics over a fixed graph. Together, these protections are part of the operational design of PRAG-\uppercase\expandafter{\romannumeral1}, while their formal leakage model and security proof are deferred to Section~\ref{sec:ap_mitigation}.

\begin{algorithm}[t]
\SetAlgoLined
\SetKwInOut{KwIn}{KwIn}
\SetKwInOut{KwOut}{KwOut}
\footnotesize
\caption{Protocol $\Pi_{\mathsf{SecureUpdate}}$}
\label{protocol:secure_update}
\KwIn{Encrypted HNSW graph $G$, new encrypted node $(\ID_{\mathsf{new}}, \ct_{\mathsf{new}})$, scaling parameter $\tau$, search depth $\text{Dep}_h$, and maximum degree $M$.}
\KwOut{Updated encrypted HNSW graph $G$.}

\smallskip
\tcp{Procedure SecureInsert}
$L \gets \mathsf{RandomLayer}()$, $\mathit{ep} \gets G.\mathsf{entryPoint}$, $\mathit{topLayer} \gets G.\mathsf{maxLayer}$\;
\For{$l = \mathit{topLayer}, \dots, L+1$}{
  $\mathcal{N} \gets \mathsf{Neighbors}(G,\mathit{ep},l)$\;
  Compute $\ct_{d_k} \gets \mathsf{HomoIP}(\ct_{\mathsf{new}}, \mathcal{N}[k])$\;
  Compute $\ct_{w_k} \gets \mathsf{ChebyApprox}(\{\ct_{d_k}\}, \tau)$\;
  $\mathit{ep} \gets \sum_k \ct_{w_k} \cdot \mathcal{N}[k]$\;
}
\For{$l = \min(L,\mathit{topLayer}), \dots, 0$}{
  $\mathcal{N} \gets \mathsf{GreedySearch}(G,\ct_{\mathsf{new}},\mathit{ep},\text{Dep}_h,l)$\;
  $\mathcal{W} \gets \mathsf{ChebyApprox}(\mathcal{N},\tau)$\;
  $\mathcal{N}' \gets \mathsf{WeightedTopM}(\mathcal{N},\mathcal{W},M)$\;
  $G.\mathsf{link}(\ct_{\mathsf{new}},\mathcal{N}',l)$\;
  $\mathit{ep} \gets \sum \mathcal{W} \cdot \mathcal{N}'$\;
}
Output updated graph $G$\;

\smallskip
\tcp{Procedure SecureDelete}
$\ct_{\mathsf{node}} \gets \mathsf{FindNodeByID}(\mathcal{I}, \ID_{\mathsf{del}})$\;
\If{$\ct_{\mathsf{node}}$ is found}{
  Mark node as deleted using a secure flag; keep it for routing\;
}
\Return deletion status\;
\end{algorithm}


\par \textbf{Update Stage: Secure Data Insertion and Deletion.} As shown in Protocol~\ref{protocol:secure_update}, this protocol supports dynamic encrypted-index updates for evolving datasets. The $\algofont{SecureInsert}$ procedure adapts standard HNSW insertion to homomorphic encryption. For a new encrypted node $(\ID_{\text{new}}, \ct_{\text{new}})$, a random maximum layer $L$ is assigned following HNSW probabilistic layering. From the top layer down to $L+1$, the protocol performs a preliminary search: from entry point $\textit{ep}$, it retrieves candidate neighbors and computes encrypted distances $\ct_{d_k}$ via homomorphic inner products. Because discrete selection is incompatible with HE, $\algofont{ChebyApprox}$ assigns Chebyshev-derived weights $\ct_{w_k}$ based on distances, with $\tau$ controlling the approximation scale. The next entry point is computed as a weighted sum of candidates, yielding a continuous approximation for descent.

From layer $L$ downward to the base layer (0), a more thorough greedy search is conducted with depth controlled by a depth heuristic, denoted as $\text{Dep}_h$ to find potential neighbors. Again, Chebyshev-derived weights are applied to these neighbors, and the top $M$ are selected using a weighted ranking $\text{WeightedTopM}$ to establish bidirectional links in the graph. The entry point for the next layer is updated as a weighted aggregate. 

The $\algofont{SecureDelete}$ procedure addresses deletion without compromising security. Rather than physically removing a node, which could disrupt graph connectivity and potentially leak information through observable changes in access patterns, the protocol locates the encrypted node by its ID and sets a deleted flag. This "mark-and-sweep" strategy allows the node to remain as a routing point, ensuring ongoing search efficiency. 

\subsection{\texorpdfstring{PRAG-\uppercase\expandafter{\romannumeral2}}{PRAG-II}}
\label{subsec:client_assisted}

To reduce approximation errors from the Chebyshev polynomial surrogate, we introduce an interactive variant of our PRAG. In this mode, the cloud computes all distance-related operations homomorphically, while the client partially decrypts intermediate encrypted distances, such as cluster distances or HNSW layer candidates, to decide the next navigation step. This design preserves the hierarchical greedy structure of HNSW and achieves retrieval accuracy closely approximating plaintext search, at the cost of multiple client-server communication rounds. The preprocessing and final RAG response phases are the same as those in the PRAG-\uppercase\expandafter{\romannumeral1}, so they will not be described again in this subsection.

\par \textbf{Interactive Setup.} This procedure builds the encrypted two-level index $\mathcal{I}$ on an untrusted server under an interactive client-assisted protocol, sharing the same overall structure as the PRAG-\uppercase\expandafter{\romannumeral1} index construction protocol $\Pi_{\mathsf{Setup}}$ but replacing Chebyshev polynomial approximations with client-guided exact decisions. The client-side encryption and upload of vectors $\mathcal{V}$ to produce ciphertexts $\mathcal{CT}_{\mathcal{V}}$ remain identical. 

On the server side, during each K-means iteration, the process diverges: instead of using $\algofont{ChebyApprox}$ for weighted assignments, the server computes exact encrypted distances from each $\ct_i$ to all centroids using $\algofont{HomoDist}$,  which contextually utilizes $\algofont{HomoIP}$ to generate encrypted similarity scores for precise client-side comparison, and sends these ciphertexts to the client. The client decrypts them, identifies the nearest cluster $j^*$ in plaintext, and sends back only the cluster index. The server then aggregates the assigned ciphertexts homomorphically and updates centroids via $\algofont{HomoNorm}$, which essentially performs a homomorphic division $ct_{sum}/ct_{count}$ using a polynomial approximation of the reciprocal function $x \mapsto 1/x$ to compute the new centroid of the aggregated vectors, mirroring the aggregation step in the PRAG-\uppercase\expandafter{\romannumeral1} but with hard assignments. 

After clustering, the construction of encrypted HNSW graphs per cluster follows the interactive insertion procedure described below, rather than the PRAG-\uppercase\expandafter{\romannumeral1}. Finally, the encrypted centroids are sent to the client for decryption into plaintext metadata $C_{\text{meta}}$, as in the PRAG-\uppercase\expandafter{\romannumeral1}. This achieves identical clustering results to plaintext K-means without approximations, while preserving privacy, at the expense of $O(NCT)$ communication rounds during indexing.

\begin{table*}[!htbp]
\centering
\caption{Comparative noise analysis of operations in encrypted vector similarity search}
\label{tab:noise_comparison}
\begin{tabular}{lccc}
\toprule
\textbf{Operation} & \textbf{Inner Product Count} & \textbf{Order of Magnitude} & \textbf{Noise Characteristics} \\
\midrule
Single inner product & $1$ & $O(1)$ & Baseline unit ($\mathcal{S}_{\text{IP}}$) \\
Single query & $C + C_{\text{probe}} \cdot L \cdot \log(M)$ & $O(10^2)$ & Independent computations \\
Single HNSW insertion & $L \cdot \text{Dep}_h$ & $O(10^2)$ & Independent comparisons \\
K-means (one iteration) & $N \cdot C$ & $O(10^6) - O(10^7)$ & \textbf{Dominant noise source} \\
Full index construction & $T \cdot \vert \mathcal{D} \vert \cdot C + \vert \mathcal{D} \vert \cdot L \cdot \text{Dep}_h$ & $O(10^7) - O(10^8)$ & \textbf{Critical bottleneck} \\
\bottomrule
\end{tabular}
\end{table*}

\par \textbf{Interactive Retrieval.} This phase extends the PRAG-\uppercase\expandafter{\romannumeral1} retrieval protocol $\Pi_{\mathsf{SecureRetrieval}}$ into an interactive protocol that removes Chebyshev polynomial approximations in favor of exact client-side ranking decisions. The client-side query preparation, including embedding $q$ into $\vect{v}_q$ and encrypting it to $\ct_q$, is unchanged. 

However, cluster routing differs: the server computes encrypted distances from $\ct_q$ to all centroids and sends them to the client, who decrypts and selects the top-$C_{\text{probe}}$ clusters in plaintext, returning their indices. This replaces homomorphic Chebyshev-weighted routing with exact client-side selection. For each selected cluster, HNSW traversal is guided interactively: at each layer, the server computes encrypted distances from $\ct_q$ to candidate nodes and sends them to the client, who decrypts, selects the nearest entry point, or at lower layers selects the top-$M$ neighbors, and returns the choices. The final layer-0 greedy search on the server produces exact top-$k$ candidates, matching plaintext HNSW behavior.

\par \textbf{Interactive Update.} To support dynamic updates, this extends the PRAG-\uppercase\expandafter{\romannumeral1} update protocol $\Pi_{\mathsf{SecureUpdate}}$ to an interactive form, replacing Chebyshev weighting with client-assisted exact selections. During insertion, layer assignment and entry point initialization remain the same. 

At each layer, the server computes encrypted distances from the new $\ct_{\text{new}}$ to candidates using $\algofont{HomoDist}$ and sends them to the client, who decrypts and selects the nearest neighbor as the next entry point or selects the top-$M$ closest nodes for linking, returning the encrypted indices. This ensures exact HNSW topology without approximations. Deletion uses the same non-interactive mark-and-sweep strategy: flagging the node as deleted while keeping it for routing, with optional offline reconstruction. Overall, this preserves the graph's structural integrity under encryption, with client-guided decisions adding communication but guaranteeing accuracy equivalent to plaintext operations.

\subsection{Operation-Error Estimation (OEE)}
\label{sec:correctness_model}
A key challenge in deploying CKKS-based retrieval systems is managing computational noise and approximation errors. As established in our contribution summary, retrieval quality hinges on preserving the relative order of similarity scores rather than achieving pointwise numeric accuracy. To formally guarantee this property, we abstract the OEE mechanism introduced earlier into a \emph{Ranking-Preserving Correctness Model}. This model provides a theoretical framework that quantifies how homomorphic noise propagates through similarity computations, and derives explicit bounds to ensure that ranking inversions occur with negligible probability under properly tuned CKKS parameters.

To make this dependency explicit, PRAG decomposes the encrypted score of candidate $d_i$ as
\[
\hat{s}_i = s_i + \xi_i^{\mathrm{poly}} + \xi_i^{\mathrm{he}} + \xi_i^{\mathrm{path}},
\]
where $s_i=\langle q,d_i\rangle$ is the true similarity, $\xi_i^{\mathrm{poly}}$ is the approximation error introduced by Chebyshev-based comparison surrogates, $\xi_i^{\mathrm{he}}$ captures CKKS arithmetic noise from inner products, rescaling, and rotations, and $\xi_i^{\mathrm{path}}$ captures path-dependent accumulation caused by different traversal depths or update paths. In PRAG-\uppercase\expandafter{\romannumeral1}, all three terms may appear; in PRAG-\uppercase\expandafter{\romannumeral2}, client-assisted comparisons eliminate the dominant comparison-approximation term for interactive decisions, leaving CKKS arithmetic noise and path-dependent accumulation as the main sources.

OEE shifts the correctness target from minimizing average numerical error to preserving the margin between competing candidates. Let $q$ be a query and $d_i, d_j$ two documents with true similarity gap $\delta = |\langle q, d_i \rangle - \langle q, d_j \rangle|$. Even if each encrypted score has error less than $\delta/2$, asymmetric noise or approximation can still flip their order. Such ranking errors break top-$k$ retrieval and harm RAG output, but are invisible to standard metrics such as mean square error (MSE). 

Ranking is preserved if, whenever the true similarity gap exceeds a margin $\Delta$ (where $\Delta$ depends on homomorphic noise and approximation error), the encrypted retrieval maintains the correct order with overwhelming probability. Concretely, for documents $d_i$ and $d_j$, ranking is stable whenever:
\[
\left|(\hat{s}_i-\hat{s}_j)-(s_i-s_j)\right| \leq \epsilon_{\mathrm{OEE}} < \Delta/2.
\]
By choosing parameters so that the total perturbation is less than $\Delta/2$, PRAG ensures that ranking errors occur only when documents are semantically indistinguishable, meaning any ambiguity reflects true embedding similarity rather than encryption-induced leakage.

OEE also reveals that noise is unevenly distributed across phases. Table~\ref{tab:noise_comparison} shows that online retrieval usually requires only hundreds of independent inner-product evaluations, whereas encrypted K-means and full index construction require millions to tens of millions. Thus, the dominant risk is not one query overflowing the noise budget; it is that offline construction determines the minimum CKKS parameter budget needed for the whole system.

This phase-aware view directly guides mitigation. At the cryptographic level, we choose a sufficiently large polynomial modulus degree $N$, an adequate ciphertext moduli chain $\{q_i\}$, and a calibrated scaling factor $\Delta_{\mathsf{scale}}$ so that the full construction circuit fits within the available noise budget. At the algorithmic level, we limit K-means to a small number of iterations and adopt conservative HNSW construction parameters such as a smaller $\text{Dep}_h$. These choices reduce multiplicative depth and operation count in the dominant offline phase while keeping retrieval quality acceptable. Additional derivations are provided in Appendix~\ref{sec:noise_analysis}.

\section{Theoretical Analysis}
We first present a formal security analysis under the honest-but-curious adversary model. Building on CKKS IND-CPA security and game-based proofs, we establish view indistinguishability and analyze data privacy, query privacy, and result privacy, including ranking confidentiality, for both operating modes, PRAG-\uppercase\expandafter{\romannumeral1} and PRAG-\uppercase\expandafter{\romannumeral2}. We then compare PRAG's theoretical complexity with representative SSR schemes in Table~\ref{tab:complexity_comparison}.

\subsection{Security Analysis}
\subsubsection{Security Definitions and Leakage Definitions}

We formalize PRAG's security along two dimensions. The first is \emph{data-level security}, which protects ciphertext contents. The second is \emph{access-pattern-level security}, which bounds the information leaked through observable traversal patterns. Let $\mathcal{A}$ denote the honest-but-curious cloud server adversary with polynomial-time capabilities.

\begin{definition}[IND-CPA Security of CKKS]
The CKKS scheme is IND-CPA secure if for any polynomial-time adversary $\mathcal{A}$, the advantage $\Adv_{\mathcal{A}}^{\text{IND-CPA}} = \left| \Pr[\mathcal{A}(\Enc(\vect{m}_0)) = 0] - \Pr[\mathcal{A}(\Enc(\vect{m}_1)) = 1] \right|$ is negligible in the security parameter $\lambda$, where $\vect{m}_0, \vect{m}_1$ are chosen by $\mathcal{A}$.
\end{definition}
This provides the foundational guarantee that all encrypted embeddings, queries, similarity scores, and rankings are semantically hidden.

\begin{definition}[Leakage Function]
The permitted leakage $\mathcal{L}$ characterizes the access-pattern information observable by the server. However, PRAG introduces \emph{random access perturbation} to suppress the fidelity of this leakage. Formally, let $\mathcal{L}_{\text{true}}$ denote the true access pattern and $\mathcal{L}_{\text{obs}}$ denote the observed perturbed access pattern. PRAG ensures that $\mathcal{L}_{\text{obs}}$ is a noisy version of $\mathcal{L}_{\text{true}}$. The perturbation is parameterized by the obfuscation rate $\rho$, which denotes the fraction of dummy accesses, and by the re-encryption epoch $E$:
\begin{equation}
\mathcal{L} = \mathcal{L}_{\text{obs}} = \mathsf{Perturb}(\mathcal{L}_{\text{true}};\, \rho, E).
\end{equation}
The perturbed leakage includes: the noisy sequence of accessed nodes in encrypted HNSW graphs and the perturbed set of probed clusters, mixed with dummy accesses.
Notably, $\mathcal{L}$ does not include plaintext values, similarity scores, or exact rankings. We do not explicitly include update-stage leakage in $\mathcal{L}$, as index updates are treated as offline operations with observable effects limited to local index connectivity.
\end{definition}

\begin{definition}[Semantic Confidentiality under Perturbed Leakage]
PRAG achieves semantic confidentiality under perturbed leakage if: (1)~\emph{Data-level}: $\mathcal{A}$ cannot distinguish plaintext embeddings, queries, or result rankings with non-negligible advantage; (2)~\emph{Access-pattern-level}: the server observes only the perturbed leakage $\mathcal{L}$ induced by random access perturbation, and any graph reconstruction attempt is therefore limited to what can be inferred from this noisy view.
Formally, for any two datasets $\mathcal{V}_0, \mathcal{V}_1$ and queries $q_0, q_1$ of equal size, the views $\View_{\mathcal{A}}(\mathcal{V}_0, q_0)$ and $\View_{\mathcal{A}}(\mathcal{V}_1, q_1)$ are computationally indistinguishable, conditioned on equivalent perturbed leakage $\mathcal{L}(\mathcal{V}_0, q_0) = \mathcal{L}(\mathcal{V}_1, q_1)$.
\end{definition}

This two-tier definition captures four properties: data privacy through IND-CPA indistinguishability of ciphertexts, query privacy through indistinguishability of query intents, result privacy through indistinguishability of scores and rankings, and access-pattern mitigation through perturbed leakage. 

To formalize how OEE masks fine-grained ranking information, let $s_i = \langle q, d_i \rangle$ be true cosine similarities and let the server observe approximations $s_i' = s_i + \xi_i$, where $|\xi_i| \leq \epsilon_{\mathrm{OEE}}$. Define the ranking ambiguity margin $\Delta := 2\epsilon_{\mathrm{OEE}}$.

\begin{lemma}[Ranking Indistinguishability under OEE]
\label{lem:oee-ranking}
For any two documents $d_i, d_j$, if $|s_i - s_j| < \Delta$, then the cloud server cannot determine their relative order with non-negligible advantage.
\end{lemma}

The lemma states that OEE converts sufficiently small score gaps into ranking ambiguity, so close candidates do not reveal stable semantic orderings to the server. The full proof is deferred to Appendix~\ref{app:security_proofs}.

\subsubsection{\texorpdfstring{Security of PRAG-\uppercase\expandafter{\romannumeral1} and PRAG-\uppercase\expandafter{\romannumeral2}}{Security of PRAG-I and PRAG-II}}
In PRAG-\uppercase\expandafter{\romannumeral1}, all operations are performed homomorphically using Chebyshev polynomial approximations to replace comparisons.

\begin{table*}[!htbp]
\centering
\caption{Complexity of Secure Similarity Search Schemes}
\label{tab:complexity_comparison}

\begingroup
\setlength{\tabcolsep}{3pt}
\begin{tabular}{@{}>{\raggedright\arraybackslash}p{0.16\textwidth}
                >{\centering\arraybackslash}p{0.19\textwidth}
                >{\centering\arraybackslash}p{0.19\textwidth}
                >{\centering\arraybackslash}p{0.18\textwidth}
                >{\centering\arraybackslash}p{0.19\textwidth}@{}}
\toprule
\textbf{Scheme} & \textbf{Setup} & \textbf{Retrieval} & \textbf{Update} & \textbf{Communication} \\
\midrule
SESR~\cite{wang2024secure} & $O(|\mathcal{D}| \cdot d \log |\mathcal{D}|)$ & $O(d \log |\mathcal{D}|)$ & NA & $O(d \log |\mathcal{D}|)$ \\
PDQ~\cite{tan2021efficient} & $O(|\mathcal{D}| \cdot d)$ & $O(|\mathcal{D}| \cdot d)$ & NA & $O(|\mathcal{D}| \cdot d)$ \\
MSecKNN~\cite{Li2025MSecKNNMS} & $O(|\mathcal{D}| \cdot d)$ & $O(|\mathcal{D}|(d + \log^2 |\mathcal{D}|))$ & $O(d)$ & $O(|\mathcal{D}|(d + \log^2 |\mathcal{D}|))$ \\
FSkNN~\cite{Fukuchi2025SecureKF} & $O(|\mathcal{D}| \cdot d)$ & $O(|\mathcal{D}|(d + k))$ & $O(d)$ & $O(|\mathcal{D}| \cdot k)$ \\

\midrule
\textbf{PRAG} & $O(|\mathcal{D}| \cdot d \log |\mathcal{D}|)$ & $O(d)$ & $O(d)$ & $O(d)$ \\
\bottomrule
\end{tabular}
\endgroup

\medskip
\noindent\parbox{\textwidth}{%
\small
\textbf{Notations:} $|\mathcal{D}|$: dataset size; $d$: vector dimension; $k$: top-$k$ retrieval parameter (FSkNN). }
\end{table*}
\begin{theorem}
Assuming CKKS is IND-CPA secure and the OEE is bounded such that ranking distortion is limited by an ambiguity margin $\Delta$, PRAG-\uppercase\expandafter{\romannumeral1} achieves semantic confidentiality with leakage $\mathcal{L}$.
\end{theorem}

\begin{proof}[Proof sketch]
The proof uses a standard hybrid sequence. We first replace real encryptions with encryptions of random vectors and then simulate homomorphic additions and multiplications on these random ciphertexts; IND-CPA security of CKKS makes both transitions negligible. We then use Lemma~\ref{lem:oee-ranking} to show that Chebyshev-driven routing decisions reveal only ambiguity-preserving order information when score gaps fall inside the OEE margin. Finally, a simulator that depends only on perturbed leakage $\mathcal{L}$ reproduces the observed paths and intermediates, so the adversary's final view is independent of the underlying plaintexts up to negligible advantage. A complete proof is given in Appendix~\ref{app:security_proofs}.
\end{proof}

PRAG-\uppercase\expandafter{\romannumeral2} uses client-assisted decryption to resolve comparisons, introducing bounded additional leakage of local orderings.

\begin{theorem}
Assuming CKKS is IND-CPA secure, PRAG-\uppercase\expandafter{\romannumeral2} achieves semantic confidentiality with leakage $\mathcal{L}$ extended to include revealed orderings of interacted candidates.
\end{theorem}

\begin{proof}[Proof sketch]
The hybrid structure is the same as in Theorem~1 for all non-interactive ciphertext operations. The difference is that the simulator must additionally reproduce the bounded local orderings revealed by client-assisted comparisons, such as selected clusters or candidate neighbors at a given HNSW layer. Because the server observes only the choices and not the decrypted values themselves, these interactions can be simulated as orderings over candidate sets whose sizes match the protocol, yielding only localized leakage and no exposure of global plaintext rankings. Thus, PRAG-\uppercase\expandafter{\romannumeral2} preserves data and query privacy while incurring only the explicitly modeled ordering leakage. The full proof appears in Appendix~\ref{app:security_proofs}.
\end{proof}

\subsubsection{Proof of Access-Pattern Mitigation}\label{sec:ap_mitigation}\label{sec:ap_proof}
While Theorems~1 and~2 establish data-level semantic confidentiality, the server can still observe retrieval paths over encrypted HNSW graphs. In PRAG-\uppercase\expandafter{\romannumeral1}, this leakage is operationally mitigated by query obfuscation with dummy traversals and by periodic re-encryption/re-clustering across epochs, as described in Section~\ref{subsec:chebyshev_based}. We now formalize the security of this mitigation and bound the adversary's graph-reconstruction advantage.

We now formally prove that the combination of query obfuscation and periodic re-encryption effectively defends against Leakage-Abuse Attacks (LAAs).

\begin{theorem}[LAA Resilience]
Let $\mathcal{A}$ be a PPT adversary that observes $Q_E$ queries within a single re-encryption epoch of length $E$. Under PRAG's access-pattern mitigation with obfuscation rate $\rho$ and epoch length $E$, the adversary's advantage in reconstructing the HNSW graph topology satisfies:
\begin{equation}
\Adv_{\mathcal{A}}^{\text{graph-recon}} \leq \frac{Q_E}{(1+\rho) \cdot |\mathcal{D}|} + \negl(\lambda).
\end{equation}
\end{theorem}

\begin{proof}[Proof sketch]
The analysis again proceeds by hybrids. First, dummy traversals are computationally indistinguishable from real traversals because both execute the same encrypted operations over ciphertext nodes; thus the adversary cannot reliably separate them. Second, periodic re-encryption and re-clustering make different epochs statistically independent, so graph evidence cannot be accumulated indefinitely across epochs. Finally, recovering an edge requires observing it inside a real traversal, whose expected count is diluted by both the dataset size and the obfuscation factor, yielding the bound $\frac{Q_E}{(1+\rho)|\mathcal{D}|}+\negl(\lambda)$. The detailed proof is deferred to Appendix~\ref{app:security_proofs}.
\end{proof}

\begin{corollary}[Epoch Length Selection]
To achieve graph reconstruction advantage at most $\alpha$ for a dataset of size $|\mathcal{D}|$, the epoch length should satisfy $E \leq \alpha \cdot (1+\rho) \cdot |\mathcal{D}|$ queries. For the 100,000-sample setting used in Section~\ref{subsec:laa_eval}, taking $\rho=0.3$ and $\alpha=10^{-2}$ gives $E \leq 1{,}300$ queries per epoch, placing the experimentally evaluated $E=1{,}000$ setting within the theoretical bound.
\end{corollary}

\begin{remark}[Security Interpretation]
Combining Theorems~1--3 with the empirical results in Section~\ref{subsec:laa_eval}, we conclude that PRAG achieves semantic confidentiality while making access-pattern-based graph reconstruction attacks difficult both theoretically and empirically.
\end{remark}

\subsection{Complexity Analysis}
In this subsection, we first analyze the computational complexity and then investigate the communication complexity.

\textbf{Setup.}
PRAG builds a two-level encrypted index with setup complexity
$O(|\mathcal{D}| \cdot d \log |\mathcal{D}|)$, covering encrypted clustering initialization and per-cluster HNSW construction. The $\log |\mathcal{D}|$ factor arises from HNSW's probabilistic layering, which requires each of the $|\mathcal{D}|$ vectors to be inserted across $O(\log |\mathcal{D}|)$ layers. As shown in Table~\ref{tab:complexity_comparison}, this is comparable to \text{SESR}, which also relies on hierarchical index construction. By contrast, \text{MSecKNN}, \text{FSkNN}, and \text{PDQ} have setup complexity $O(|\mathcal{D}| \cdot d)$, because they avoid hierarchical ANN index construction. PRAG therefore incurs higher upfront indexing cost but enables substantially more efficient sublinear retrieval.


\textbf{Retrieval.} As summarized in  Table~\ref{tab:complexity_comparison}, PRAG achieves sublinear retrieval complexity $O(d)$ by performing encrypted hierarchical navigation over cluster-level centroids followed by localized HNSW graph traversal. In contrast, SESR has complexity $O(d \log |\mathcal{D}|)$ with logarithmic dependency on the full dataset size and lacks support for encrypted similarity ranking within graph-based ANN structures. More importantly, PDQ, MSecKNN, and FSkNN remain effectively linear in the dataset size during query processing, with costs $O(|\mathcal{D}| \cdot d)$, $O(|\mathcal{D}|(d + \log^2 |\mathcal{D}|))$, and $O(|\mathcal{D}|(d + k))$, respectively. This reflects the fact that they evaluate secure distance computations over the outsourced dataset without a hierarchical ANN pruning mechanism, making them less suitable than PRAG for large-scale RAG retrieval.

\textbf{Update.} PRAG supports efficient incremental updates with complexity $O(d)$, independent of the dataset size $|\mathcal{D}|$, by enabling localized encrypted insertions into existing HNSW graphs. In contrast, SESR and PDQ do not support real-time updates in the compared settings. MSecKNN and FSkNN have record-level update costs of $O(d)$ because they do not maintain a hierarchical encrypted index, but this advantage comes with a structural trade-off: query processing continues to scale linearly with the outsourced dataset rather than benefiting from localized ANN navigation.

\textbf{Communication.} PRAG incurs communication overhead of only $O(d)$ per query, as it transmits a single encrypted query vector and returns encrypted similarity scores. SESR increases communication cost to $O(d \log |\mathcal{D}|)$ to support unlinkability through multi-path traversal. By comparison, PDQ requires $O(|\mathcal{D}| \cdot d)$ communication, MSecKNN requires $O(|\mathcal{D}|(d + \log^2 |\mathcal{D}|))$, and FSkNN requires $O(|\mathcal{D}| \cdot k)$, since these schemes exchange information proportional to the outsourced dataset during secure query evaluation. This makes PRAG considerably more communication-efficient for cloud-scale RAG workloads.
\section{Evaluations}
\label{sec:experiments}

\textbf{Experimental Setup.} We evaluate PRAG on a 100,000-sample subset of TriviaQA~\cite{joshi2017triviaqa}, which serves as the default dataset scale throughout this section. All experiments are conducted on an Ubuntu 24.04.2 LTS server equipped with 8 Intel Xeon Gold vCPUs, 200~GB RAM, and 500~GB SSD. For secure computation, we adopt the CKKS scheme with a 128-bit security level. Unless otherwise specified, retrieval and communication use a top-10 setting, and the graph-reconstruction study uses $Q=10{,}000$ queries.

\textbf{Implementation.} The PRAG prototype is a pure C++ implementation developed on top of the Microsoft SEAL 4.1 library. The source code are available at https://github.com/richikun2014-bit/PRAG. We utilize LangChain for initial document chunking and OpenAI\_Emb for embedding generation, with Qwen-3 serving as the response generation model\footnote{\url{https://huggingface.co/Qwen/Qwen3-32B-GGUF}}. A key technical feature of our implementation is the use of Chebyshev polynomials specifically for encrypted-domain numerical comparisons; this mechanism facilitates secure decision making during HNSW hierarchical navigation and routing path selection.


\subsection{System Performance}
\label{subsec:complexity_analysis}
We evaluate PRAG in both PRAG-\uppercase\expandafter{\romannumeral1} and PRAG-\uppercase\expandafter{\romannumeral2} configurations against baselines that satisfy currently acceptable security assumptions, focusing on retrieval efficiency, update cost, communication overhead, and resilience to graph reconstruction attacks. Unless otherwise specified, the retrieval and communication results reported below use a top-10 retrieval setting. Following this criterion, Huang et al.~\cite{huang2019toward} and SPE-Sim~\cite{zheng2020achieving} are excluded from experimental baselines due to known security weaknesses.

\textbf{Setup.} FSkNN at $25.902$~s and MSecKNN at $44.678$~s show the fastest setup, as illustrated in Fig.~\ref{fig:setup_retrieval}(a). This advantage mainly comes from their lightweight preprocessing pipelines, which avoid costly tree/graph index construction and rely on simpler encryption/partition routines. In contrast, PRAG-\uppercase\expandafter{\romannumeral2} at $289.2956$~s, PRAG-\uppercase\expandafter{\romannumeral1} at $398.0$~s, PDQ at $217.524$~s, and SESR at $986.385$~s incur higher upfront construction overhead due to more complex secure indexing procedures. Nevertheless, this additional setup cost in PRAG is amortized by substantially lower retrieval and communication overhead in the online phase, especially compared with heavier alternatives such as PDQ and SESR.

\textbf{Retrieval.}
PRAG-\uppercase\expandafter{\romannumeral1} achieves the best retrieval latency at $1.29$~s, followed by PRAG-\uppercase\expandafter{\romannumeral2} at $7.91$~s, as shown in Fig.~\ref{fig:setup_retrieval}(b). Both outperform SESR at $166.11$~s, PDQ at $4107.111$~s, MSecKNN at $2137.65$~s, and FSkNN at $1978$~s. PDQ is substantially slower because it executes a full comparison circuit entirely in the ciphertext domain. MSecKNN and FSkNN are also slower because secure ranking requires many cross-server interaction rounds. Overall, PRAG achieves the most favorable retrieval efficiency among the compared secure schemes.

\begin{figure}[!htbp]
  \centering
  \begin{minipage}[b]{0.48\columnwidth}
    \centering
    \includegraphics[width=\textwidth]{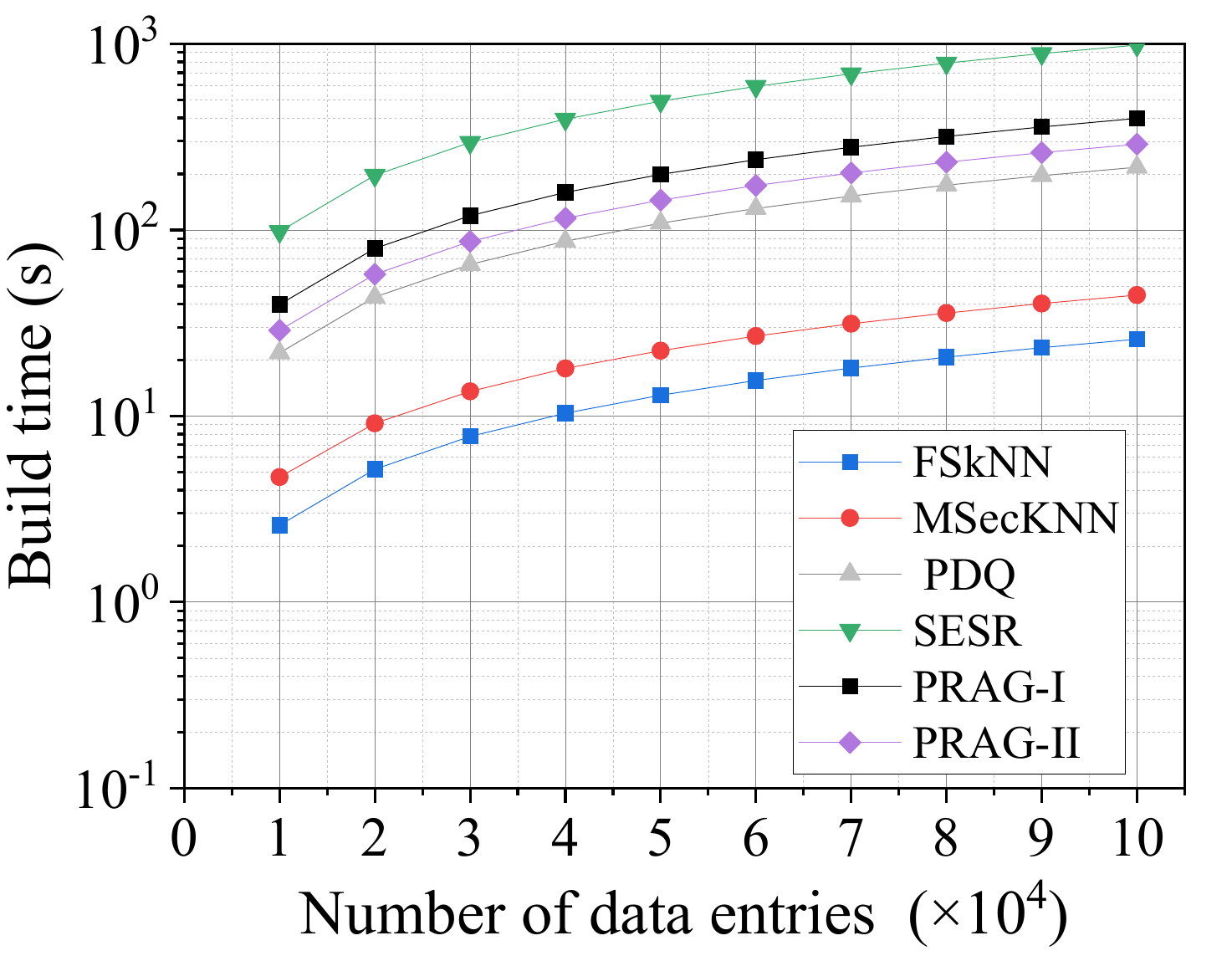}
    \centerline{\footnotesize (a) Setup time}
  \end{minipage}
  \hfil
  \begin{minipage}[b]{0.48\columnwidth}
    \centering
    \includegraphics[width=\textwidth]{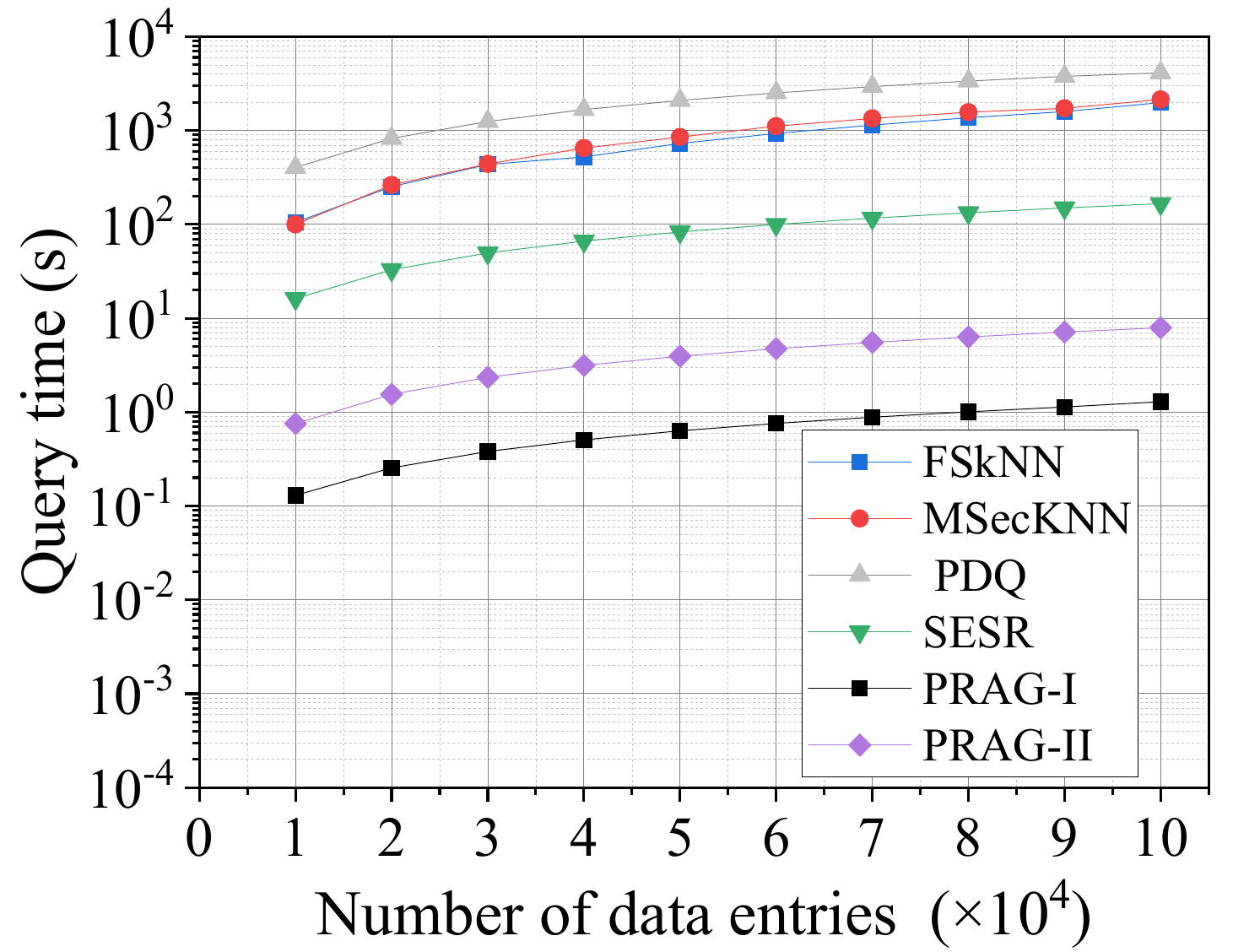}
    \centerline{\footnotesize (b) Retrieval time}
  \end{minipage}
  \caption{Setup time and retrieval time comparison across schemes.}
  \label{fig:setup_retrieval}
\end{figure}

\textbf{Update.}
PRAG-\uppercase\expandafter{\romannumeral2} achieves $5.44$~ms and PRAG-\uppercase\expandafter{\romannumeral1} achieves $7.28$~ms for updates as shown in Fig.~\ref{fig:update_commu}(a), while FSkNN and MSecKNN are about $7.42$~ms and $12.198$~ms. PDQ and SESR do not support updates. 

\begin{figure}[!htbp]
  \centering
  \begin{minipage}[b]{0.48\columnwidth}
    \centering
    \includegraphics[width=\textwidth]{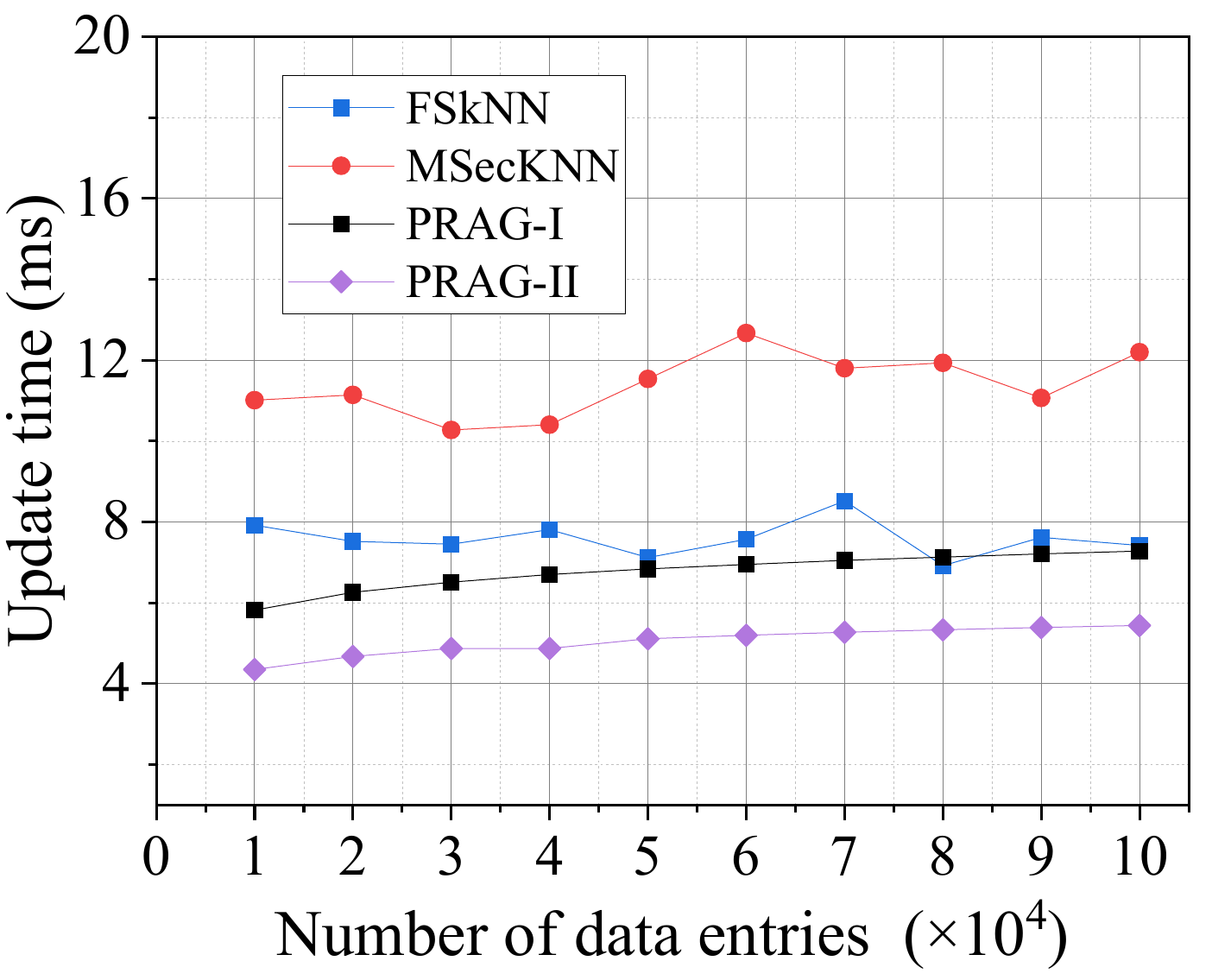}
    \centerline{\footnotesize (a) Update time}
  \end{minipage}
  \hfil
  \begin{minipage}[b]{0.48\columnwidth}
    \centering
    \includegraphics[width=\textwidth]{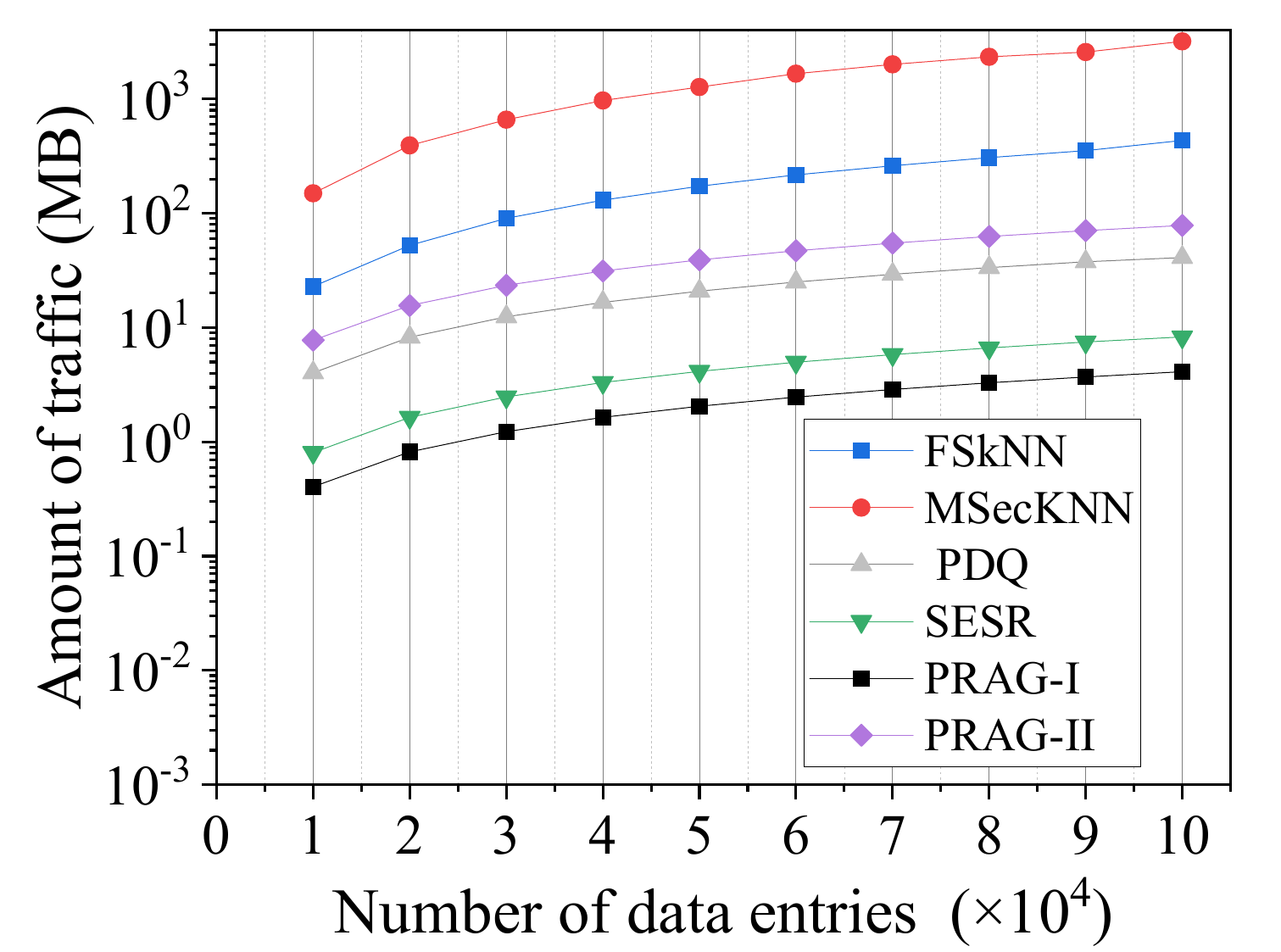}
    \centerline{\footnotesize (b) Communication cost}
  \end{minipage}
  \caption{Update time and communication cost comparison across schemes.}
  \label{fig:update_commu}
\end{figure}

\textbf{Communication.}
PRAG-\uppercase\expandafter{\romannumeral1} achieves the lowest communication cost at $4.1175$~MB, as shown in Fig.~\ref{fig:update_commu}(b). SESR and PDQ follow at $8.3055$~MB and $41.07111$~MB, while PRAG-\uppercase\expandafter{\romannumeral2} reaches $78.5478$~MB due to layer-wise client-assisted interactions during encrypted traversal. FSkNN and MSecKNN are much higher at $434.7$~MB and $3206.475$~MB, respectively, because their secure ranking procedures require substantially more cross-party message exchanges. Overall, PRAG-\uppercase\expandafter{\romannumeral1} provides the most communication-efficient retrieval path, and PRAG-\uppercase\expandafter{\romannumeral2} preserves stronger ranking fidelity with moderate but still practical bandwidth overhead.

PRAG offers the strongest end-to-end trade-off across setup, retrieval, update, and communication. Although its offline setup cost is higher than FSkNN and MSecKNN, it substantially reduces online retrieval latency and bandwidth, which dominate continuous RAG serving. Compared with PDQ and SESR, PRAG also avoids severe update bottlenecks. Within PRAG, PRAG-\uppercase\expandafter{\romannumeral1} is preferable for latency- and bandwidth-sensitive deployments, while PRAG-\uppercase\expandafter{\romannumeral2} is preferable when stronger ranking fidelity is required with still-practical communication overhead. This dual-mode design improves deployability over single-regime alternatives in cloud RAG workloads.

\subsection{Resilience Against Graph Reconstruction Attacks}\label{subsec:laa_eval}

We evaluate PRAG's access-pattern mitigation (Section~\ref{sec:ap_mitigation}) by simulating an IHOP-style leakage-abuse adversary~\cite{oya2022ihop,damie2021highly,kornaropoulos2021state} that reconstructs the encrypted HNSW graph from observed traversal paths. We report edge recovery rate, false positive rate, and Recall@10 degradation under obfuscation rates $\rho \in \{0,0.1,0.2,0.3,0.5\}$ and re-encryption epochs $E \in \{\infty,5000,2000,1000\}$.

\begin{table}[!t]
\centering
\caption{Adversary's Graph Reconstruction Success under PRAG's Access-Pattern Mitigation.}
\label{tab:laa_defense}
\small
\begin{tabular}{cccccc}
\toprule
\textbf{$\rho$} & \textbf{Epoch $E$} & \makecell{\textbf{Edge}\\\textbf{Recovery}} & \makecell{\textbf{False}\\\textbf{Positive}} & \makecell{\textbf{Recall}\\\textbf{@10}} & \makecell{\textbf{Latency}\\\textbf{Overhead}} \\
\midrule
0    & $\infty$  & 34.7\% & 8.2\%  & 72.45\% & 1.00$\times$ \\
0.1  & $\infty$  & 18.3\% & 21.6\% & 72.31\% & 1.10$\times$ \\
0.2  & $\infty$  & 11.5\% & 35.4\% & 72.18\% & 1.19$\times$ \\
0.3  & $\infty$  & 6.8\%  & 48.7\% & 71.92\% & 1.28$\times$ \\
0.5  & $\infty$  & 3.1\%  & 62.3\% & 71.54\% & 1.47$\times$ \\
\midrule
0    & 5000      & 19.2\% & 14.5\% & 72.45\% & 1.00$\times$ \\
0    & 2000      & 9.6\%  & 28.3\% & 72.45\% & 1.00$\times$ \\
0    & 1000      & 5.1\%  & 41.7\% & 72.45\% & 1.00$\times$ \\
\midrule
0.3  & 2000      & 2.4\%  & 71.8\% & 71.92\% & 1.28$\times$ \\
0.3  & 1000      & 1.1\%  & 83.5\% & 71.92\% & 1.28$\times$ \\
\bottomrule
\end{tabular}
\begin{tablenotes}
  \small
  \item[] Epoch $E=\infty$ denotes no periodic re-encryption.
\end{tablenotes}
\end{table}

Table~\ref{tab:laa_defense} shows that unmitigated access leakage is substantial: at $\rho=0$ and $E=\infty$, the adversary recovers 34.7\% of true edges. Either defense already helps: query obfuscation with $\rho=0.3$ lowers recovery to 6.8\% at a modest Recall@10 drop from 72.45\% to 71.54\% and 1.28$\times$ latency overhead, while periodic re-encryption with $E=1{,}000$ lowers recovery to 5.1\% without affecting retrieval quality. When both are enabled ($\rho=0.3$, $E=1{,}000$), recovery falls to 1.1\% and the false positive rate rises to 83.5\%, indicating that the inferred graph is dominated by spurious edges. Overall, PRAG's mitigation sharply limits graph reconstruction while preserving retrieval utility, consistent with Theorem~3.

\section{Conclusion and Future Work}
We present PRAG, a privacy-preserving RAG framework for untrusted clouds that combines dual-mode encrypted retrieval, OEE-based ranking stabilization, and access-pattern mitigation. Experiments and formal analysis show that PRAG achieves practical latency, efficient updates, competitive recall, and stronger resilience to graph reconstruction attacks while preserving end-to-end confidentiality. PRAG currently assumes a single-client key setting. Future work will extend it to multi-key secure computation and further reduce homomorphic overhead through algorithmic and hardware acceleration.

\normalem
\bibliographystyle{IEEEtran}

\bibliography{ref}


\appendices

\section{Leakage Function}\label{leakage}

\begin{definition}[Leakage Function (with Access-Pattern Perturbation)]
\label{def:leakage}
The leakage function $\mathcal{L} = (\mathcal{L}_{\text{setup}}, \mathcal{L}_{\text{query}})$ characterizes the \emph{perturbed} side-channel information observable by the adversary after applying PRAG's access-pattern mitigation (Section~\ref{sec:ap_mitigation}):

\textbf{Setup Phase:}
\begin{equation}
\mathcal{L}_{\text{setup}}(\mathcal{D}) = \left(|\mathcal{D}|, d, C, \left\{\mathcal{C}_j^{(t)}\right\}_{j=1,t=1}^{C,T}\right)
\end{equation}
where $|\mathcal{D}|$ is the dataset size, $d$ is the embedding dimension, $C$ is the number of clusters, and $\mathcal{C}_j^{(t)}$ captures cluster assignment patterns at iteration $t$. Note that setup leakage is invalidated after each re-encryption epoch.

\textbf{Query Phase (Perturbed):}
\begin{equation}
\mathcal{L}_{\text{query}}(\mathbf{q}^{(i)}, \mathcal{D}) = \left(\widetilde{\mathcal{J}}_{\text{probe}}^{(i)}, \left\{\widetilde{\mathsf{Path}}_j^{(i)}\right\}_{j \in \widetilde{\mathcal{J}}_{\text{probe}}^{(i)}}, K\right)
\end{equation}
where $\widetilde{\mathcal{J}}_{\text{probe}}^{(i)} = \mathcal{J}_{\text{probe}}^{(i)} \cup \mathcal{J}_{\text{dummy}}^{(i)}$ includes both real probed clusters and dummy cluster probes, $\widetilde{\mathsf{Path}}_j^{(i)}$ includes real traversal paths padded to fixed length $\ell_{\max}$ and interspersed with dummy random walks, and $K$ is the result set size. The adversary cannot distinguish real paths from dummy paths.

Crucially, $\mathcal{L}$ does NOT include: plaintext embeddings $\{\mathbf{v}_i\}$, queries $\{\mathbf{q}^{(i)}\}$, similarity scores $\{\langle \mathbf{q}^{(i)}, \mathbf{v}_j \rangle\}$, exact ranking relationships, or the distinction between real and dummy accesses.
\end{definition}

\section{Supplementary Security Proofs}
\label{app:security_proofs}

\subsection{Proof of Lemma~\ref{lem:oee-ranking}}
\begin{proof}
The encrypted similarity difference satisfies
\[
s_i' - s_j' = (s_i - s_j) + (\xi_i - \xi_j).
\]
By the triangle inequality,
\[
|(\xi_i - \xi_j)| \le 2\epsilon_{\mathrm{OEE}} = \Delta .
\]
Hence, when $|s_i - s_j| < \Delta$, the sign of $s_i' - s_j'$ is dominated by encryption noise and approximation error. As a result, the probability that $s_i' > s_j'$ is negligibly close to the probability that $s_j' > s_i'$.

Therefore, the relative ordering of $d_i$ and $d_j$ is statistically indistinguishable from a random permutation from the adversary's perspective.
\end{proof}

\subsection{Proof of Theorem 1}
\begin{proof}
We prove via a sequence of hybrid games that the adversary's advantage in breaking semantic confidentiality is negligible.

\noindent\textbf{Game 0:} Real execution. $\mathcal{A}$ observes $\mathcal{CT}_{\mathcal{V}} = \{\Enc(\vect{v}_i)\}$, encrypted centroids, HNSW graphs, and intermediates from Chebyshev-based clustering and retrieval.

\noindent\textbf{Game 1:} Replace all encryptions with encryptions of random vectors $\vect{r}_i$. By IND-CPA of CKKS, the difference in $\mathcal{A}$'s success probability between Game 0 and Game 1 is negligible.

\noindent\textbf{Game 2:} Simulate homomorphic operations on random ciphertexts. Since homomorphic additions and multiplications preserve IND-CPA, and the Chebyshev surrogate is a polynomial function, intermediates remain indistinguishable.

\noindent\textbf{Ranking Hiding via OEE-Induced Ambiguity.} By Lemma~\ref{lem:oee-ranking}, whenever two candidates fall within the $\Delta$-margin, the cloud server cannot determine their relative order with non-negligible advantage.

\noindent\textbf{Implication for Access Patterns.} All routing decisions in clustering and HNSW traversal are driven by Chebyshev-derived weights computed from $\{s_i'\}$. By Lemma~\ref{lem:oee-ranking}, whenever competing candidates fall within the $\Delta$-margin, these routing decisions behave as randomized choices. Consequently, the observed access patterns are statistically independent of the true semantic ranking, up to negligible probability.

\noindent\textbf{Game 3:} Simulate access patterns using a leakage simulator $\Sim_{\mathcal{L}}$ that generates paths based solely on index size and structure, such as random cluster probes and graph traversals matching the real distribution. Since real paths depend only on approximated Chebyshev-derived weights rather than exact rankings, and OEE ensures that noise and approximation bounds $\Delta$ make close rankings ambiguous, the simulated paths are statistically close to real ones. The difference in $\mathcal{A}$'s success probability between Game 2 and Game 3 is negligible.

In Game 3, $\mathcal{A}$'s view is independent of plaintexts, proving data and query privacy. Result privacy follows from ranking ambiguity: since the Chebyshev surrogate smooths decisions, exact orderings are not revealed, and $\Delta$ bounds inferable semantic relationships.

Thus, $\Adv_{\mathcal{A}}^{\text{sem-conf}} \leq \negl(\lambda)$.
\end{proof}

\subsection{Proof of Theorem 2}
\begin{proof}
The proof follows a similar game sequence as in Theorem 1. The key difference is that during clustering, insertion, and retrieval, the client decrypts limited intermediates to resolve comparisons, revealing exact local orderings such as the selected cluster and the nearest neighbor.

\noindent\textbf{Game 0:} Real PRAG-\uppercase\expandafter{\romannumeral2} execution. During clustering and insertion, for each distance set $\{\ct_{d_{ij}}\}$, the client decrypts and returns exact assignments or links. For retrieval, per layer, the client decrypts candidate distances and returns exact nearest or top-$M$ choices. $\mathcal{A}$ observes these choices, such as the selected cluster $j^*$ and the entry points.

\noindent\textbf{Game 1:} As in PRAG-\uppercase\expandafter{\romannumeral1} Game 1, replace encryptions with random vectors. By IND-CPA of CKKS, $|\Pr[\mathcal{A} \text{ wins Game 0}] - \Pr[\mathcal{A} \text{ wins Game 1}]| \leq \negl(\lambda)$.

\noindent\textbf{Game 2:} Simulate homomorphic computations such as $\algofont{HomoDist}$. Identical to PRAG-\uppercase\expandafter{\romannumeral1} Game 2: the difference is 0.

\noindent\textbf{Game 3:} Simulate access patterns and interacted orderings with $\Sim_{\mathcal{L}}$. For non-interacted parts, use the simulation from PRAG-\uppercase\expandafter{\romannumeral1}. For interactions, since $\mathcal{A}$ does not observe decrypted distances and observes only the resulting choices, simulate choices as random orderings over candidate sets whose sizes match the real protocol, such as size $C$ for clusters and layer-dependent sizes for HNSW. Leakage is bounded: per query, at most $O(L)$ orderings are revealed, where $L$ is the maximum number of HNSW layers, and these orderings are localized to small candidate sets. Since candidates are random from Game 1, revealed orderings do not correlate with global plaintext rankings, and thus $|\Pr[\mathcal{A} \text{ wins Game 2}] - \Pr[\mathcal{A} \text{ wins Game 3}]| \leq \negl(\lambda)$.

In Game 3, the adversary's view is independent of plaintexts except for explicit, bounded ordering leakage, proving data and query privacy. Result privacy holds because the final top-$k$ set is computed client-side and never revealed to the server, while intermediate ciphertexts hide similarity scores.

The non-interactive encryption components are identical to those in PRAG-\uppercase\expandafter{\romannumeral1}, so their security follows directly. Therefore, $\Adv_{\mathcal{A}}^{\text{sem-conf}} \leq \negl(\lambda)$.
\end{proof}

\subsection{Proof of Theorem 3}
\begin{proof}
We prove via a sequence of games that the adversary's graph reconstruction advantage is bounded.

\noindent\textbf{Game 0 (Real Execution with Mitigation):} The adversary observes the perturbed access patterns $\mathsf{Path}_{\text{obs}}$ for $Q_E$ queries within a single epoch. Each observation consists of real traversal paths mixed with a $\rho$-fraction of dummy paths.

\noindent\textbf{Game 1 (Indistinguishability of Real and Dummy Paths):} We replace all real query traversal paths with independent random walks. Since CKKS ciphertexts are IND-CPA secure, the computational operations performed during real traversals and dummy traversals are indistinguishable to the server. Specifically, both real and dummy accesses involve reading encrypted nodes, computing homomorphic inner products, and returning encrypted results. Without the secret key, the server cannot determine whether a traversal follows the greedy HNSW routing or a random walk. The distinguishing advantage between Game 0 and Game 1 is bounded by $\negl(\lambda)$ via a standard hybrid argument over the $Q_E \cdot (1+\rho) \cdot \ell_{\max}$ ciphertext operations.

\noindent\textbf{Game 2 (Independent Epochs):} Since each epoch uses fresh keys, fresh random seeds for clustering, and fresh HNSW level assignments, the access patterns in epoch $e$ are statistically independent of the graph structure in epoch $e' \neq e$. Thus, the adversary's observation is limited to $Q_E$ queries per epoch.

\noindent\textbf{Graph Reconstruction Bound:} Consider the adversary's task of recovering an edge $(u,v)$ in the HNSW graph. An edge can only be detected if both $u$ and $v$ appear consecutively in a real non-dummy traversal path. Within a single epoch, the probability that a specific edge is traversed by a random query is at most $\frac{1}{|\mathcal{D}|}$. With $Q_E$ observed queries, each containing $(1+\rho)$-diluted paths, the expected number of observations of any specific edge is
\begin{equation}
\mathbb{E}[\text{obs}(u,v)] \leq \frac{Q_E}{(1+\rho) \cdot |\mathcal{D}|}.
\end{equation}
For LAA reconstruction to succeed with high confidence, the adversary needs $\Omega(|\mathcal{D}|)$ observations per edge~\cite{kellaris2016generic}. By choosing epoch length $E$ such that $Q_E \ll (1+\rho) \cdot |\mathcal{D}|$, the adversary's reconstruction advantage remains negligible.

Combining the three games, we obtain
\[
\Adv_{\mathcal{A}}^{\text{graph-recon}} \leq \frac{Q_E}{(1+\rho) \cdot |\mathcal{D}|} + \negl(\lambda).
\]
\end{proof}


\section{Protocols Supplement}\label{algorithm}

\subsection{Homomorphic Normalization (\algofont{HomoNorm})}
Let $\ct_{\mathsf{sum}}$ denote a ciphertext encrypting a vector 
$\vect{u} \in \mathbb{R}^d$, and let $\ct_{\mathsf{count}}$ denote a ciphertext
encrypting a positive scalar $c \in \mathbb{R}_{>0}$.
The homomorphic normalization operation, denoted by $\algofont{HomoNorm}$, computes
\[
\ct_{\mu} \gets \algofont{HomoNorm}(\ct_{\mathsf{sum}}, \ct_{\mathsf{count}})
= \mathsf{Enc}\!\left(\frac{\vect{u}}{c}\right),
\]
using a polynomial approximation of the reciprocal function $x \mapsto 1/x$.

\subsection{Additional Protocol}
\begin{algorithm}[!htbp]
\SetAlgoLined
\SetKwInOut{KwIn}{KwIn}
\SetKwInOut{KwOut}{KwOut}
\footnotesize
\caption{Protocol $\Pi_{\mathsf{LocalRAG}}$}
\label{protocol:LocalRAG}
\KwIn{Aggregated retrieval results, cluster scores, original query $q$, local database $DB_{\text{local}}$, and target size $K$.}
\KwOut{Generated response $R$.}
$\mathit{doc\_score} \gets \emptyset$\;
\ForEach{$(j^\ast,\mathit{PartialResults})$ in aggregated results}{
  $\mathit{cluster\_weight} \gets \mathit{cluster\_score}[j^\ast]$\;
  \ForEach{$\ID_i$ at position $\mathit{rank}$ in $\mathit{PartialResults}$}{
    $\mathit{rrf\_score} \gets 1/(60+\mathit{rank})$\;
    $\mathit{doc\_score}[\ID_i] \gets \mathit{doc\_score}[\ID_i] + \mathit{rrf\_score}\cdot\mathit{cluster\_weight}$\;
  }
}
Select top-$k$ document identifiers from $\mathit{doc\_score}$ as $\mathit{Candidates}$\;
$\{\ID_1^\ast,\ldots,\ID_K^\ast\} \gets \mathsf{FuseRerank}(q,\mathit{Candidates},DB_{\text{local}},K)$\;
$\mathit{Context} \gets \text{""}$\;
\ForEach{$\ID_i^\ast$ in selected identifiers}{
  $d_i^\ast \gets \mathsf{GetLocalData}(DB_{\text{local}}, \ID_i^\ast)$\;
  $\mathit{Context} \gets \mathit{Context} \oplus \mathsf{ExtractText}(d_i^\ast)$\;
}
$R \gets \mathsf{LLM.Gen}(\mathit{Prompt})$\;
\Return $R$\;
\end{algorithm}


\section{Supplementary OEE Analysis}
\label{sec:noise_analysis}
\noindent This appendix complements the main-text discussion in Section~\ref{sec:correctness_model} by giving the phase-wise derivation behind the noise counts and mitigation rationale.

\subsection{Understanding CKKS Noise Characteristics}

Before comparing noise across different operational phases, we first characterize how fundamental homomorphic operations contribute to noise growth in CKKS. Homomorphic addition exhibits linear, gradual noise growth: $\text{Noise}(\text{ct}_1 + \text{ct}_2) \approx \text{Noise}(\text{ct}_1) + \text{Noise}(\text{ct}_2)$. In contrast, homomorphic multiplication is the primary driver of noise accumulation, with quadratic behavior: $\text{Noise}(\text{ct}_1 \times \text{ct}_2) \approx C \cdot \text{Noise}(\text{ct}_1) \cdot \text{Noise}(\text{ct}_2)$. Rescaling operations, required after multiplication to maintain numerical stability, introduce minor noise while increasing relative noise. Rotation operations contribute moderate noise, typically greater than addition but less severe than multiplication.

A single homomorphic inner product $\langle \text{ct}_q, \text{ct}_{\text{db}} \rangle$ between two encrypted $d$-dimensional vectors requires one component-wise multiplication, one relinearization operation for key switching, one rescaling operation, and approximately $\log(d)$ rotations and additions for summation. We denote the noise introduced by this complete circuit as our basic unit $\mathcal{S}_{\text{IP}}$, which serves as the fundamental measure for comparing computational costs across different phases.

\subsection{Comparative Analysis of Noise Sources}

The total noise in our system originates from two primary phases: the online retrieval phase and the offline index construction phase. To understand which phase dominates the noise budget, we analyze the computational complexity of each in terms of the number of noise-intensive homomorphic operations.

\subsubsection{Noise in the Retrieval Phase}

A single query execution involves a sequence of independent $\algofont{HomoIP}$ computations. The total number of such operations consists of two components. First, during cluster pruning, the encrypted query vector is compared with all $C$ encrypted cluster centroids, incurring $C \times \mathcal{S}_{\text{IP}}$ of noise spread across $C$ independent computations. Second, during HNSW search within the top $C_{\text{probe}}$ selected clusters, a typical greedy traversal from the top layer to the bottom involves approximately $L \cdot \log(M)$ inner product operations per cluster, where $L$ denotes the number of layers and $M$ represents the maximum node degree.

The overall computational cost can be expressed as:
\begin{equation}
    \text{Query}_{\text{cost}} \approx C + C_{\text{probe}} \cdot L \cdot \log(M) \cdot \ct_{s_j}
\end{equation}

Crucially, the noise from these operations does not accumulate sequentially. Each similarity score is computed independently. The primary risk here is that the noise in any single $\algofont{HomoIP}$ computation could be large enough to flip the ranking between two closely scored vectors. However, the total computational load for a single query remains relatively modest, typically in the order of hundreds of operations.

\subsubsection{Noise in the Index Construction Phase}

In stark contrast, the offline index construction phase represents a greater computational burden and is the dominant source of noise concern. This phase comprises two major components, each with notably different noise profiles.

\textbf{K-means Clustering Noise}
The K-means algorithm is the most computationally intensive component. Each iteration has two steps. In the assignment step, every data point is compared with all cluster centroids, requiring $N \cdot C$ independent $\algofont{HomoIP}$ computations per iteration. For example, with a database of $10^5$ vectors and $100$ clusters, one iteration requires $10^7$ inner-product evaluations. The update step computes cluster means through $(\vert \mathcal{D} \vert/C - 1)$ additions per cluster, followed by scalar multiplication for normalization. Although less noisy than assignment, it still contributes to the overall budget.

With $T$ iterations, the total K-means cost is:
\begin{equation}
    \text{K-means}_{\text{cost}} = T \cdot \vert \mathcal{D} \vert \cdot C \cdot \ct_{s_j}
\end{equation}

\textbf{HNSW Construction Noise}
Building the HNSW graph structure requires inserting each of the $\vert \mathcal{D} \vert$ data points into the hierarchical graph. Each insertion involves a layer-wise greedy search from top to bottom, with neighbor selection at each layer using search depth $\text{Dep}_h$. This process requires approximately $L \cdot \text{Dep}_h$ inner products per insertion. The cumulative construction cost is:
\begin{equation}
    \text{HNSW}_{\text{cost}} = \vert \mathcal{D} \vert \cdot L \cdot \text{Dep}_h \cdot \ct_{s_j}
\end{equation}

\subsubsection{Quantitative Comparison}

Table~\ref{tab:noise_comparison} compares noise across system operations and highlights a clear disparity in computational burden. The dominant sources are K-means iterations and full index construction. Since full index construction is unavoidable, mitigation should prioritize reducing noise from K-means iterations.

\subsubsection{Conclusion of Analysis}

The analysis shows that index construction, especially iterative K-means clustering, dominates the required noise budget. The K-means assignment step alone executes 4--5 orders of magnitude more inner products than a single query. The resulting volume of homomorphic multiplications in this offline phase strongly influences the minimum cryptographic parameters needed for correct system operation. CKKS parameters should therefore be sized for offline construction, not for the comparatively light online query phase. Once sized for construction, the query phase remains within budget. This implies that mitigation should focus primarily on noise generated during index construction.



\end{document}